\documentclass[journal]{IEEEtran}
\usepackage[cmex10]{amsmath}
\usepackage{algorithm,algorithmic,amsbsy,amsmath,amssymb,epsfig,bbm,mathrsfs,amsthm,dsfont}
\usepackage{rotating}
\usepackage{multirow}
\usepackage{booktabs}
\usepackage{cite}
\usepackage{graphicx}
\usepackage{subfigure}
\usepackage{array}
\usepackage{bm}
\usepackage{pifont}
\usepackage{bbding}
\newtheorem{theorem}{Theorem}
\newtheorem{lemma}{Lemma}
\newtheorem{proposition}{Proposition}

\newtheorem{definition}{Definition}
\allowdisplaybreaks

\renewcommand{\algorithmicensure}{\text{BSs Perform:}}  

\begin{document}

\title{Decentralized Energy Allocation for Wireless Networks with Renewable Energy Powered Base Stations}
\author{\authorblockN{Dapeng Li, Walid Saad, Ismail Guvenc, Abolfazl Mehbodniya, and Fumiyuki Adachi}
\vspace{-0.8cm}
\thanks{Dapeng Li is affiliated with the Key Lab of Broadband Wireless
 Communication and Sensor Network Technology of the Ministry of Education, Nanjing University of Posts and Telecommunications, China, email: dapengli@njupt. edu.cn.}
\thanks{Walid Saad is with Wireless@VT, Bradley Dept. of Electrical and Computer Engineering, Virginia Tech, USA,
and is also with  Dept. of Computer Engineering, Kyung Hee University, South Korea, email: walids@vt.edu.}
\thanks{Ismail Guvenc is with Dept. of Electrical and Computer Engineering, Florida International  University, USA, email: iguvenc@fiu.edu.}
\thanks{Abolfazl Mehbodniya and Fumiyuki Adachi are with Wireless Signal Processing and Networking Lab, Dept. of Comm. Engineering, Tohoku University, Japan, emails: mehbod@mobile.ecei.tohoku.ac.jp, adachi@ecei. tohoku.ac.jp. Dr. Saad was a corresponding author.}}
\maketitle

\begin{abstract}
In this paper, a green wireless communication system in which base stations are powered by renewable energy sources is considered. This system consists of a capacity-constrained  renewable power supplier (RPS) and a base station (BS) that faces a predictable random connection demand from  mobile user equipments (UEs). In this model, the BS powered via a combination of a renewable power source and the conventional electric grid, seeks to specify the renewable power inventory policy, i.e., the power storage level. On the other hand, the RPS must strategically choose the energy amount that is supplied to the BS. An M/M/1 make-to-stock queuing model is proposed to investigate the decentralized decisions when the two parties optimize their individual costs in a noncooperative manner.
The problem is formulated as a noncooperative game whose Nash equilibrium (NE) strategies are  characterized in order to identify the causes of inefficiency in the decentralized
operation. A set of simple linear contracts are introduced to coordinate the system so as to achieve an optimal system performance. The proposed approach is then extended to a setting with one monopolistic RPS and $N$ BSs that are privately informed of their optimal energy inventory levels.
In this scenario, we show that the widely-used proportional allocation mechanism is no
longer socially optimal. In order to make the BSs truthfully report their energy demand, an incentive compatible (IC) mechanism is proposed for our model.  Simulation results show that using the green energy can present significant traditional energy savings for the BS when the connection demand is not heavy. Moreover, the proposed scheme provides valuable energy cost savings by allowing the BSs to smartly use a combination of renewable and traditional energy, even when the BS has a heavy traffic of connections. Also, the results show that performance of the proposed IC mechanism will be close to the social optimal, when the green energy production capacity increases.
\end{abstract}

\begin{keywords}
Green communications, renewable power supply, game theory, contract, mechanism design.
\end{keywords}

\IEEEpeerreviewmaketitle

\section{Introduction}\vspace{-0.35em}
\IEEEPARstart{E}nergy efficiency has recently emerged as a major research challenge in the next generation of wireless systems \cite{Yu} in order to reduce the carbon footprint and $\text{CO}_2$ emission of wireless networks. The electric bill has become a significant
portion of the operational expenditure of cellular operators, and
the $\text{CO}_2$ emission produced by wireless cellular networks
are equivalent to those from more than 8 million cars \cite{Krishnamachari}.
The largest fraction of power consumption in wireless networks comes
from base stations (BSs), especially, when they are deployed in large numbers\cite{Marsan}.
With this premise, power saving in BSs is particularly important for network operators.
In order to reduce the energy consumption of the BS, an online algorithm is used to enable BSs to switch between on/off states according to traffic characteristics \cite{Koutitas1}.
At the same time, the role of renewable energy generation  will be a promising energy
alternative for future mobile networks. Understanding
the interaction between the random green power generations and the dynamics of the energy consumption on wireless networks becomes a main challenge facing future green communications design.

Manufacturers and network operators such as Ericsson, Huawei, Vodafone and China Mobile have started developing the BS with a renewable power source~\cite{Gozalvez,Belfqih,Huawei}. Relating to to the operation of modern radio and data center
networks, the concepts such as demand response, supply load control, and the model of
the prosumer in the smart grid are explored in \cite{Koutitas2}.
However, solar energy and wind energy are not controllable generation resources like traditional generation resources such as coal or natural gas. These resources are intermittent and can have a random output. It is, hence, both desirable and challenging to design and optimize the green energy enabled mobile networks. The recent work in \cite{Ansari} lays out basic design principles and research challenges on optimizing the green energy powered mobile networks, and points out that green energy powered BSs should be properly designed and optimized to cope with the dynamics of green power and mobile data traffic.

Considering the use of renewable energy sources, the authors in \cite{FZhang} study throughput and value based wireless transmission scheduling under time limits and energy constraints for wireless
network. The BSs' transmission strategies that can be optimized to reduce the energy demands without degrading the quality of service of the network are investigated in \cite{Ho}.  The authors in \cite{Han1} propose a packet scheduling algorithm shaping the BS's energy demands to match the green power generation. In order to sustain traffic demand of all users, green energy utilization is optimized by balancing the traffic loads among the BSs \cite{Han2}.
The green energy source aware user association schemes are examined in \cite{HKim}.
These research attempts mainly focus on adapting BSs' transmission strategies by using the green energy sources. However, most of these works do not account for energy storage nor do they investigate prospective energy storage strategies.

Since mobile traffic shows temporal dynamics, a BS's energy
demands change over time. Instead of focusing on specific techniques, the authors in \cite{Ansari} provide guidelines showing that BSs could determine how much energy is utilized at the current stage and how much energy is reserved for the future.
The authors in \cite{Farbod} propose to reduce the BS's power consumption at certain stages and reserve energy for the future to satisfy the network's outage constraint. The recent work in \cite{Niyato} provides a stochastic programming formulation to minimize the BS's energy storage cost (i.e., the battery self-discharge cost) and the cost of using the renewable energy and electric energy. However, there are many issues that these works do not tackle and remain to be addressed. For instance, the renewable source and the BS generally belong to different operators. As a result, to deploy renewable-powered BSs, it is important to understand the interaction/competition between the renewable energy supplier and the BS, especially considering the competition's impact on the QoS.

The main contribution of this paper is to propose a novel noncooperative game model to investigate the optimized energy allocation strategies of the renewable power supplier (RPS) and the BS.
On the one hand, the RPS incurs supply cost and the possible QoS performance reducing cost. On the other hand, the BS incurs energy reservation cost and also the QoS cost.
They will unilaterally choose their supply and reservation strategies, respectively, to minimize their individual cost.We formulate the problem as a noncooperative game between the RPS and the BS.

To solve the studied game, we propose an M/M/1 make-to-stock queue model to analyze
the competitive behavior between the RPS and the BS. Different from the traditional queue,
the make-to-stock queue has a buffer of resource laying in between the end user and the server.
In our queue model, the RPS can be viewed as a server. The order for energy storage replenishment from the BS will be placed to the RPS server. We explicitly characterize the NE strategies
and present a set of simple linear contracts to coordinate the system to optimize the overall network performance and achieve an efficient point, under a decentralized design.
We then extend the result to a setting with one monopolistic limited capacity RPS and $N$ BSs.
BSs are privately informed of their optimal energy inventory levels. If the energy orders of a given BS exceed the available capacity, the RPS allocates capacity using a publicly known allocation mechanism,
i.e., a mapping from BS orders to capacity assignments. Based on our model, an incentive compatible mechanism that induces BSs truthfully
report their energy demands is proposed.
In summary,  we make the following contributions,
\begin{itemize}
\item We analyze the decentralized energy allocation optimization for wireless networks with
green the energy powered BS.  Using an M/M/1 make-to-stock queue model, we investigate how the supply rate of RPS and the energy storage/inventory level affect the QoS.

\item We then model and analyze the interactions between the RPS and a BS by using a noncooperative game.  We prove the existence and uniqueness
of the equilibrium, and show how various system
parameters (i.e., the QoS reducing cost and the cost splitting factor) affect
the equilibrium behavior.

\item Based on the NE solution, we investigate how the BS controls the use of the combination of renewable energy and the traditional energy. We then explore the centralized optimal system performance and  propose a set of simple linear contracts to overcome the inefficiency of the NE.

\item We study the energy allocation mechanism design for a system
in which a limited capacity monopolistic RPS provides energy to multiple BSs.
We propose a truth-inducing mechanism where BSs truthfully report their optimal energy demand is a dominant equilibrium.
\end{itemize}

To the best of our knowledge, little work has been done for developping decentralized
algorithms that allow to capture the renewable energy allocation strategies,
particularly when multiple BSs operate. One interesting open problem is  to study how the predictable traffics affect the BSs' adaptive energy management strategies and how
the network performance is affected by forecast renewable production capacity, as done in this paper.

The rest of the paper is organized as follows. The system model and
the noncooperative game formulation are presented in Section II. In Section III, we analyze the game and
prove the existence and uniqueness of the NE, and we propose a linear contract to coordinate the system.
Section IV examines the renewable energy allocation mechanism design for a setting with
one limited capacity RPS and multiple BSs. We provide numerical results
and discussion in Section V. Section VI concludes the paper.
\vspace{-0.5em}
\section{System Model} \label{sec:mod}
We start by describing, in detail, the different entities of the studied system.  The description of the competition for renewable power inventories
and power supply capacity between a BS and the RPS is then presented.
\vspace{-0.5em}
\subsection{System Components}
\vspace{-0.5em}
\subsubsection{Electric grid and renewable power supplier}
In our model, we use the term electric grid to refer to a controllable generation resource such as coal,
natural gas, nuclear, or hydro. Existing research works, i.e., \cite{Ansari}, often make a similar assumption to investigate the energy allocation from the controllable electric grid or dynamic green energy source for BSs.
The power supplied from the electric grid has a price per unit of quantity, per unit of time.
We assume that in a certain \textit{production period},
the supplier output power is a random variable with mean $\mu_0$
(unit of energy per unit of time)\footnote{For example,  the authors in \cite{Hoff} investigated an approach to estimate the standard deviation of the change in output of solar energy over some time interval (such as one minute), using data taken from some time period (such as one year).}.
From the BS's perspective, the RPS's production facility is modeled as a
single-server queue with service times that are exponentially
distributed with rate $\mu$ which is referred to as the renewable energy supply rate.
The RPS is responsible for choosing
the parameter $\mu$ ($\mu<\mu_0$), which accounts for the fact that the RPS has a finite capacity.

\begin{figure}[t]
\centering
\includegraphics[width=2.3in]{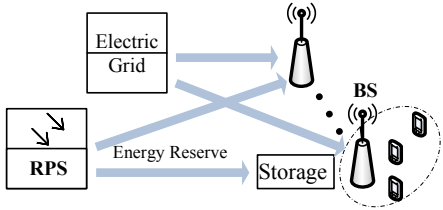}
\vspace{-0.2cm}
\caption{The considered system in which the BS can use the controllable energy from
electrical grid and it can reserve unstable renewable energy for serving UEs.}\label{sys}
\end{figure}
\vspace{-0.1em}
\subsubsection{Renewable Energy Powered Base Station}
We consider a single base station part of a wireless system that is used to provide mobile and wireless access. The energy consumptions of the BS include the energy coming from different components such as power amplifier, signal processing unit, antenna, and cooling\cite{Arnold}.
The static power consumption is constant when the
base station is active and does not have any radio transmissions.
We assume that the static component, i.e., the power consumption of the rectifiers, is served by the stationary conventional energy source. For example, the work in \cite{Deruyck} shows that the rectifiers could consume $100$~W for a microcell BS.
In contrast, the dynamic power consumption includes the digital signal processing, the
transceiver and the power amplifier and, thus, it can fluctuate during time due to
variations of the load on the base station.
The amount of power consumption depends on the type of base stations \cite{Deruyck}. In this paper, we take a typical microcell base station as an example and use the data from \cite{Niyato}, i.e., the dynamic power consumption coefficient is $24$~W (Joule/s) per connection.

The connection demand of the wireless base station depends on the usage condition,
and can be predicted from the usage history.
Indeed, modeling the arrival of a new wireless call or
message arrivals in packet-data network as a Poisson arrival process is extensively
studied \cite{Huang,Epstein,Celik}.
Moreover, research works such as in \cite{Niu1} showed that traffic demands are highly predictable. Although even for the queuing systems whose arrival or service rates have general distributions, the heavy traffic condition will generate mean queue lengths that coincide with M/M/1 queue results  \cite{Halfin}.

We model the connection demand as a homogeneous Poisson process with rate $\lambda$.
The incoming connections are assumed to form a single waiting line that depends on the order of arrival, i.e., the first-come first-served discipline. We assume that the BS consumes a unit transmission energy, $E$, for each connection in the downlink. Each connection generates a revenue per unit time to the BS.

In order to serve the connections by using the intermittent renewable energy, the BS should charge the green energy to an energy storage unit.  If the energy storage is less than a certain level,
the BS will replenish the green energy from the RPS.
The BS will pay  a price $P_1$ to the RPS for reserving each unit energy.
The BS's power reservation strategy needs to be varied dynamically and periodically as renewable
production condition changes. Hereinafter, we consider one such \textit{production period}.
Notice for the ideal case that the BS can be supplied with enough green power, there is no need to analyze green energy allocation strategies and the QoS cost again.

The BS incurs a reservation cost for holding energy in the storage unit.
Such a reservation cost includes both physical and financial components.
For example, the stored energy can decrease even without consumption. Such a
physical component is referred to as the self-discharge
phenomenon \cite{Niyato}. On the other hand, the financial holding
cost could be proportional to the energy's market price, i.e., the BS can sell the reserved energy to
other consumers  instead of holding it in the storage \cite{Tham}.
The future smart grid is able to integrate and exchange different energy flows among various users through the on-gird physical and cyber infrastructures \cite{Fang}.  Particularly, the cyber infrastructure which can perform energy trading consists of a large number of communication and computing networks, wide-area monitors, various sensors, and control functions together with necessary information processing functions.  Thus, the BS or any other users that are on the grid will have the capability to sell the reserved energy to other consumers instead of holding it in the storage. Moreover, some electricity companies have already offered initial energy buy-back programs \cite{buy-back}.

Both the self-discharge cost and the financial holding
cost can be proportional to the average  energy reservation
level. As a result,  we use the term of $c\cdot I_s$ per unit time to evaluate the energy reservation cost, where $I_s$ is the average  energy reservation level and $c$ is the cost coefficient.
Note that, increasing the energy reservation results in less backlogged connections, but  yields
a higher reservation cost.

\subsection{Renewable Energy Supply-Inventory Game}
When a connection enters the queue of the BS, it will be provided access to the spectrum only if its energy requirement amount is available in the storage.
The energy consumption of each connection is one unit, since we assume that each connection lasts one time unit. If there is not enough energy in storage, the UE has to wait until the BS reserves sufficient energy. We say that such a connection experiences a \textit{backlogged} access.
The BS should make a decision on the desired energy reservation level $s$.
If the energy storage is less than $s$ units,
the BS will place an order of one unit energy from the RPS\footnote{ As shown in in \cite{Krieger},the length of charging pulses ranging from milliseconds to a second, are scaled to correlate with the electrochemical response times in the batteries.  Also, the charging efficiency which determines the percent of energy lost during charge achieves to about $90\%$. }.
Such a decision means that the BS should initially charge the energy storage to $s$ units. Otherwise
when game begins, the BS will place energy replenishment orders no matter there are connections arrivals or not. Hence, the energy reservation level is also referred to as energy \textit{inventory} strategy.

The service process of a queuing system is characterized by the distribution of the time serving the arrival of the traffic pertaining to a given customer. In this regard, the scheduling process allows one to control the service that is allocated to the traffic classes \cite{queue}.
With regards to the RPS, its output is a random variable and has the capability to generate an average of $\mu_0$ units energy per unit time in a certain time duration.  At the same time, each arriving radio connection will consume one unit energy. The supply rate $\mu$ of the RPS implies that an average of $\mu~(\mu\leq\mu_0)$ units energy can be scheduled to serve the connections in unit time. Thus, the energy supply of RPS can be modeled as the service process of a queue. Moreover, the stock can be introduced to improve the QoS, i.e., to reduce the average queue length.
If multiple connection requests arrive at the same time, multiple energy orders will be placed to the RPS at the same time. A standard queue system will serve the orders according to a First-in-First-out principle.

The value of dynamic energy consumption is related to the accessed connections. It can be determined by the arriving accessing requests and the information from the resource scheduling process. Generally, the radio access request information is transmitted over the public control channel by mobile users \cite{Mishra}.
Then, at the beginning of each unit time,  the MAC layer of the BS can obtain the instantaneous energy demand value, and the energy supply requests can be placed to the RPS accordingly.
The key metric for a queuing system is the stationary average queue length which is directly related to the total QoS cost.  Accordingly, we mainly investigate this statistical metric instead of using the
 the instantaneous state of the queuing system in our model.

If a connection that is at the head of the queue finishes the access,
the on-hand reservation decreases by one.
Thus, an access request is equivalent to an energy reservation request, and is placed to the RPS at the time instant when a connection arrives as well, i.e., at each epoch of the demand process.
The backlogged access might generate detrimental consequences
on the system. The QoS will be deteriorated, the UE will wait and thus will have a higher delay.
To model this QoS degradation, each backlogged access
will be assigned a cost $b$ for the system.  This cost is split between the RPS and BS,
with a fraction $\alpha\in[0,1]$ charged to the BS. The parameter $\alpha$ is exogenously specified
in our model. Its value will depend on a variety of factors, such as the structures of
the market and the UEs' expectations.

We assume that the system state, the connection demand rate, and the cost parameters
are known by each agent.  The BS and the RPS select the renewable energy base reservation $s$
and the energy supply rate $\mu$, respectively, in order to maximize their own profits.
The competition between the BS and the RPS is shown in Fig. \ref{gamemodel}.
Let $D$ define the average number of backlogged accesses, and $I_s$
be the average energy storage level per unit time.
$X(\mu)$ represents the operation cost for the RPS with power supply rate $\mu$. Recall that
$c$ is the reservation cost factor. Then, the average cost per unit time for the BS is given by:
\begin{equation}\label{Cc1}
C_o(s,\mu) = c\cdot I_s + \alpha b \cdot D,
\end{equation}
and the average cost per unit time for RPS will be:
\begin{equation}\label{Co1}
C_r(s,\mu) = (1-\alpha)b\cdot D + X(\mu).
\end{equation}

\begin{figure}[t]
\centering
\includegraphics[width=3.0in]{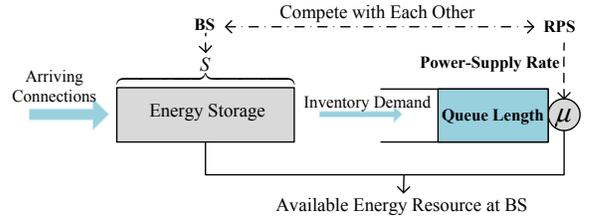}
\vspace{-0.2cm}
\caption{The renewable energy supply-inventory game model in which the energy supply-inventory strategies can be viewed as a kind of energy allocation strategies. }\label{gamemodel}
\end{figure}

Our model studies the scenario in which the RPS can sell its energy to different
entities (for example the BS and the grid). This is captured in our model by considering the changes in the load factor of the RPS. This load factor change will
generate the operation cost for the RPS and is represented by
$X(\mu)$ in the utility function.

In the next section, we will use an M/M/1 make-to-stock queue to show how $I_s$ and $D$ are determined by $s$ and $\mu$. Here, we just describe the basic model and principles.
Because the unsatisfied connection request is backlogged, and we use an average cost criterion, it follows that the agents' revenues are independent of their costs.
For example,  if the RPS supplies energy to the BS at a fixed
wholesale price $w$ (per unit energy in a time unit) and the average number of
accesses is $\lambda$ in one unit time, then the average revenue of the RPS is $w\cdot\lambda$ per unit of time which is irrelevant to its strategy $\mu$~\footnote{If the UE departs and never returns, that potential revenue is lost. This will lead to a totally different analysis, and be left
for future work.}.

Thereafter, profit maximization and cost minimization lead to the same
solution. For convenience, we adopt a cost-minimization framework.
The RPS and the BS choose their strategies (i.e., they must choose their strategy
before observing the other agent's strategy)
with the objective to optimize their individual energy reservation-related costs.
Given the model, the problem is formulated as a \emph{noncooperative strategic game}\cite{Saad} in which : 1) the \emph{players} are the BS and the RPS, 2) the \emph{strategy} of the BS and the RPS are $s$ and $\mu$, respectively, and 3) the \emph{cost functions} are given by the following equations (\ref{Cos2}) and (\ref{Cr2}).

\section{Game and Equilibrium Analysis}
To analyze this game, we first obtain the average costs for the two players for any given
pair of strategies $(s,\mu)$. We then prove that the game has a unique Nash equilibrium solution.
Finally, we compare the centralized optimal solution
with the Nash solution.
\vspace{-0.35em}
\subsection{Cost Analysis}
In an attempt to isolate and hence understand
the impact of energy inventory for the BS, we
consider a queuing model with an attached reservation: the mobile connection (traffic transmission) demands arrive at a \textit{rate} of $\lambda$ connection demands/unit time and are served by the green energy in the storage; At the time instant when a connection arrives, the BS place an energy replenishment order to the RPS which satisfies the order with a rate $\mu$.
In this regard, the radio access network request can be viewed as being served at a \textit{rate} of $\mu$ units energy/unit time. Consequently, the energy replenishment operation at the BS behaves as a single-server  queue whose service rate is $\mu$ and demand rate is $\lambda$. Such a queuing policy which is illustrated in Fig.~\ref{queue}, constitutes the central theme of the subsequent cost analysis.

Assuming one wireless connection consumes one unit energy in one unit time
is a common assumption in existing research works \cite{Ansari,Deruyck,Niyato}. Here, we clarify how to decide the unit energy and unit time. In wireless networks, the transmission time interval (TTI, several milliseconds) is a parameter that refers to the encapsulation of data from higher layers into frames for transmission on the radio link layer, thus related to the duration of a transmission on the radio link.
For example, the radio frame of an LTE system lasts 10 milliseconds \cite{lte}.
Each connection could account for a number of $t_d$ TTIs. The dynamic power consumption coefficient is $E$ W (Joule/s) per connection. Thus, each connection demand consumes $t_d\cdot\text{TTI}\cdot E$ Joule in unit time. In other words, the time unit is set to be $t_d\cdot\text{TTI}$ and each user may generate several connection demands in a certain time duration.
The value of $t_d\cdot\text{TTI}$ can range from a few milliseconds to a few seconds as it is related to
QoS cost coefficient, since the number of backlogged connections are calculated over these time intervals. Without loss of generality, in simulations, the time unit $t_d\cdot\text{TTI}$ is set to be $1$ second. For example, the energy unit can be $24$ Joule.

\begin{figure}
\centering
\includegraphics[width=2.5in]{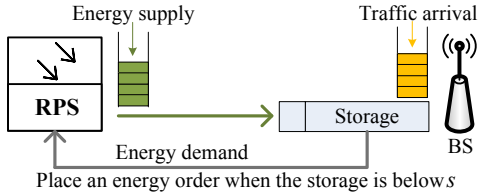}
\vspace{-0.2cm}
\caption{The proposed make-to-stock queue which has an energy storage laying in between the end mobile user and the RPS.}\label{queue}
\end{figure}

Such a queuing policy which is illustrated in Fig.\ref{queue},
constitutes the central theme of the subsequent cost analysis.
Let $p_j$ be the probability that the BS has $j$ connections in its queue at steady state.
Based on the classical birth-death process in queuing theory \cite{queue},
we have, for the statistical equilibrium, $p_j = \rho^{j}p_0, j=1,2,3,...$ where
$p_0 = 1-\rho$, and the condition for a stable queuing system is given by $\rho =
\frac{\lambda}{\mu} <1$ where $\rho$ represents the load factor.
Let $N_q$ be the stationary number of connections waiting at the queue.
Then, $N_q$ is geometrically distributed with mean $\frac{\rho}{1-\rho}$.

To simplify our analysis, we assume that $N_q$ is a continuous random variable when $N_q\geq s$,
and replace the geometric distribution with an exponential distribution with parameter
$\nu =\frac{1-\rho}{\rho}=\frac{\mu-\lambda}{\lambda}$. This continuous-state approximation can be justified by a heavy
traffic approximation, i.e., the traffic of arriving connections, and generates mean queue lengths that
coincide with M/M/1 results for all server utilization levels \cite{Halfin}. The heavy traffic approximation allows
the incorporation of general inter-arrival time and service
time distributions~\footnote{Note that the number of servers is asymptotically negligible after normalization. Moreover, in the non exponential case, $\nu$ would be divided by one-half of
the sum of the squared coefficients of variation of the
inter-arrival and service time distributions, $(c_{ia}^2+c_{is}^2)/2$.}.
In queueing theory,  an M/M/1 queue represents the queuing system having a single server, where arrivals are determined by a Poisson process and job service times have an \textit{exponential distribution}.
The average queue length depends on the load factor, which is the ratio of the service demand rate and the service supply rate. The heavy traffic condition means that the queuing system has a large load factor, which corresponds to the scenario that there is not enough green energy generation for wireless networks. For example, future BSs may request green energy from the small capacity energy supplier, i.e., the on-roof home green energy equipments. Moreover, our formulation is also suitable to model small base stations operating in rural areas, which may not have continuous access to the green energy supplier with a high capacity.

Hence, under such a condition, no matter what kind of distribution its output follows, we can use an exponential distribution to analyze the RPS's random output.
Although this continuous-state approximation may lead to slightly different
quantitative results (the approximation tends to underestimate the optimal discrete energy reservation levels), it has no effect on obtaining an insight into the system, which is the objective of our study.
Notice, the make-to-stock queuing model is essentially a queuing system which can be used to examine scenarios in which customers arrive for a given service, wait if the service cannot start immediately and leave after being served. The stock model is just introduced to reduce the average waiting length. In other words, the storable property of the energy can enable us to design the storage strategies for using the green energy in the queueing service system.

Next, we consider the average power supply cost of the RPS, $X(\mu)$ in equation (\ref{Co1}).
Recall that the maximum power production rate of the RPS is $\mu_0$.
The RPS provides the power for the BS as a production process that has exponentially
distributed inter-production times with rate $\mu$. Then, the remaining supply capacity of RPS
can be modeled as a Possion process with rate $\mu_0-\mu$. Thus, if there are power supply requests from other entities (i.e., other base stations or home appliances) with rate $\lambda_0$, their load factor will be increased by:
\begin{equation*}
\triangle_{\mu} = \frac{\lambda_0}{\mu_0-\mu}-\frac{\lambda_0}{\mu_0}=\frac{\lambda_0\mu}{\mu_0(\mu_0-\mu)}.
\end{equation*}

Then, $X(\mu) = c_s \triangle_{\mu}$ represents the operation cost for the RPS.
The average queue length depends on the load factor, which is the ratio of the service demand rate and the service supply rate.   Thus, the load factor is crucial to a queuing system. If the load factor approaches $1$, then the average queue length will go to infinity. As a result, the QoS of the queue system will be totally destroyed.
For a stationary energy source with infinite capacity, it is reasonable to assume that it will not
bear any QoS cost after deployment. However, a random energy source has a direct impact on
the QoS, which may then impair the user's experience of the network operator. This will be detrimental to the operator who may now decide to choose another energy supplier. As a result, the original supplier will lose revenues. In this regard, such an impact will be a major concern in green energy markets.
Hence, we consider that the RPS will also incur certain QoS cost in the system.

Without loss of generality, we normalize
the expected variable cost per unit time by dividing
it by the reservation cost rate $c$. Toward this end, we
normalize the cost parameters as follows:
\begin{equation*}
\tilde{b} = \frac{b}{c},~~\tilde{c}_s = \frac{c_s}{c}\frac{\lambda_0}{\mu_0},~~\tilde{c} = \frac{c}{c} = 1.
\end{equation*}
To ease the notation, hereinafter, we omit the tildes from
these parameters. There is a one-to-one correspondence between $\nu$ (we call it as normalized energy supply rate) and the
the RPS's decision variable $\mu$. Hence, the steady-state expected normalized variable cost
per unit time for the RPS and BS can be investigated in terms of the two decision variables $s$ and $\nu$.
\begin{figure}
\centering
\includegraphics[width=3.0in]{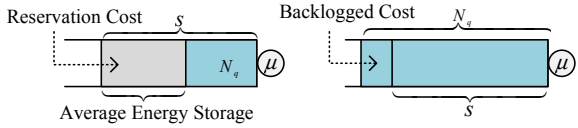}
\vspace{-0.2cm}
\caption{Illustration of renewable energy reservation cost and the backlogged cost.}\label{queuecost}
\end{figure}

\begin{proposition}
The BS's average cost can be expressed by:
\begin{equation}\label{Cos2}
C_o(s,\nu) =  s - \frac{1-e^{-\nu s}}{\nu} + \alpha b\frac{e^{-\nu s}}{\nu},
\end{equation}
while the RPS's average cost is:
\begin{equation}\label{Cr2}
\begin{split}
C_r(s,\nu) &= (1-\alpha)b\frac{e^{-\nu s}}{\nu}  + c_s\frac{\nu+1}{\varphi-\nu},
\end{split}
\end{equation}
\end{proposition}
where $\varphi = \frac{\mu_o}{\lambda} - 1$.
\begin{proof}
With continuous-state approximation, the expected energy reservation level of the BS can be expressed as,
\begin{equation*}
I_s = E[(s-N)^+] = \int_0^s(s-x)\nu e^{-\nu x}{\rm d}x = s - \frac{1-e^{-\nu s}}{\nu}.
\end{equation*}

The expected backlogged users is,
\begin{equation}\label{D}
D = E[(N - s)^+] = \int_s^\infty(x-s)\nu e^{-\nu x}{\rm d}x = \frac{e^{-\nu s}}{\nu}.
\end{equation}
Substituting the above equations into (\ref{Cc1}) yields  (\ref{Cos2}).
Substituting $X(\mu)$ and  (\ref{D}) into
(\ref{Co1}) yields (\ref{Cr2}).
\end{proof}

The backlogged cost and the reservation cost associated with the queue model are illustrated in Fig.\ref{queuecost}.
The system operates on a periodic basis, i.e., several hours during which
the energy generation capacity and connection demand rate can be predicted.
At the beginning of each period, the renewable power production capacity and the
access demand can be predicted from the history record and atmospheric parameters
which are available from monitoring devices or sensors in the control and monitoring system of a
smart grid environment~\cite{Huang}.
Notice, our model assumes fixed energy prices and is a \textit{one-shot} game. Thus, no matter whether the RPS has a storage or not, it must satisfy the energy requests from the BS rather than store the energy. If the RPS decides to sell its electricity to different time while accounting for future QoS impacts and different energy prices,
a totally new \textit{repeated game} model over time horizon should be designed.
This interesting extension is currently  not investigated here, and will be subject of future research.

The choice of a strategy $s$ by the BS corresponds to placing an energy replenish order to the RPS.
Thus, at the beginning of the game, the BS should have charged its energy storage to
an inventory of $s$ energy units.
The source of the charge can be the electric gird or the RPS.
When the game begins, the RPS supplies the energy with the determined
rate $\nu$ and the BS serves the UEs with the renewable energy,
thereby striking a balance between reducing the energy bill and maintaining the QoS.
The operation of the system is shown in Fig.~\ref{operation}.
\begin{figure}[t]
\centering
\includegraphics[width=3.0in]{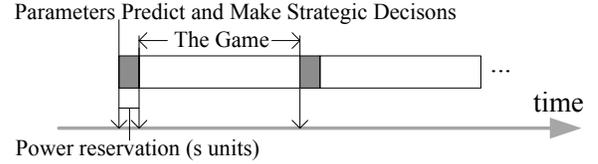}
\vspace{-0.2cm}
\caption{Illustration of the operation of the system.}\label{operation}
\end{figure}

\subsection{System Performance and Game Solution}
In this subsection, we will characterize
the Nash equilibrium (NE) strategies \cite{Osborne} and identify the causes of inefficiency in the decentralized operation.  Toward this end, we first prove the existence and
uniqueness of the NE and then compare the Nash solution with the centralized optimal solution.
Subsequently, we present a set of simple linear contracts to coordinate the system.
Finally, we discuss the adaptive power management of the BS based on the NE solution.

\subsubsection{The Nash Equilibrium of the Game}\
In a decentralized game-theoretic formulation, the BS and the RPS select their individual strategies $s$ and $\nu$ in order to minimize their own cost functions.
In other words, the BS will choose $s$ to minimize $C_o(s,\nu)$, assuming that the
RPS chooses $\nu$ to minimize $C_r(s,\nu)$; likewise, the RPS will simultaneously
choose $\nu$ to minimize $C_s(s,\nu)$, assuming the BS chooses $s$ to minimize $C_o(s,\nu)$.
A pair of strategies $(s^*,\nu^*)$ is an NE if neither the BS nor the
RPS can gain from a unilateral deviation from these strategies, i.e.,
\begin{equation}\label{best-r}
\begin{split}
&s^*(\nu)= \arg\min\limits_x C_o(x,\nu^*),\\
&\nu^*(s)= \arg\min\limits_x C_r(x,s^*).
\end{split}
\end{equation}

There are two important issues regarding the determination of the NE: existence
and uniqueness. We will show that a NE always exists in the game and the uniqueness
of the equilibrium is always guaranteed.  First,
let us define an auxiliary function in terms of the BS's backlogged
share $\alpha$,  backlogged cost $b$, and RPS's supply cost $c_s$:
\begin{equation*}
\begin{split}
f = \sqrt{\frac{b-\alpha  b+(b-\alpha b )\ln(1+\alpha b)}{c_s(1+\alpha b)}}.
\end{split}
\end{equation*}

The backlogged cost split factor $\alpha\leq 1$, thus $f$ is always a real value.
The function $f$ is used to make the Nash solution concise and will play a prominent role in our analysis.
\begin{theorem}\label{nash-t}
The proposed noncooperative game admits a unique NE,
\begin{equation}\label{nash-s}
{\nu}^* = \frac{f\varphi}{\sqrt{1+\varphi}+f},~~{s}^*= \frac{(\sqrt{1+\varphi}+f)\ln(1+\alpha b)}{ f\varphi}.
\end{equation}
\end{theorem}
\begin{proof}
$s^*(\nu)$ is the BS's reaction curve representing the optimal energy reservation level given a supply rate $\nu$.
As (\ref{Cos2}) is concave in $s$, $s^*(\nu)$ is characterized by the first-order condition
\begin{equation}\label{nashs-1}
\nu s^*(\nu) = \ln(1+\alpha b).
\end{equation}

Using a similar argument and (\ref{nashs-1}), we find that the RPS's
reaction curve $\nu^*(s)$ satisfies,
\begin{equation}\label{nashv-2}
\frac{\nu^2}{\nu-\varphi} = \frac{f}{c_s(1+\varphi)}.
\end{equation}

The unique solution to equations (\ref{nashs-1}) and (\ref{nashv-2}) yields ~(\ref{nash-s}).
\end{proof}

From Theorem \ref{nash-t}, the resulting equilibrium costs $C_o^*$  and $C_r^*$ are,
\begin{equation*}\label{nash-s-c}
\begin{split}
&C_o^* = C_o({s}^*, {\nu}^*) = {s}^*,\\
&C_r^* = C_r({s}^*, {\nu}^*) = \frac{b-\alpha b}{(1+\alpha b){\nu}^*} + \frac{s^*f}{\ln(1+\alpha b)}.
\end{split}
\end{equation*}

Since the existence of the N.E. is guaranteed, so the traditional
best response dynamics \cite{Osborne} can be used to converge to the NE. For the proposed game model,
we have got the analysis solution of the NE. Consequently, the RPS and the BS are able to
make a decision about the strategies at the beginning of the game.

Because the function $f$  is decreasing in $\alpha$ and is
increasing in $b$ for $b >0$, it follows that as $\alpha$ increases,
the BS becomes more concerned with backlogged UEs
and increases his energy reservation level ($f\downarrow \rightarrow s^*\uparrow$), while the RPS
cares less about backlogged users and builds less energy supply
rate $\nu^*$ ($f\downarrow \rightarrow \nu^*\downarrow$).

\subsubsection{Distributed algorithm to achieve the NE}
Generally, for the scenarios that the analysis solution of the NE cannot be obtained,
finding distributed algorithms to reach a NE is both a challenging and important task.
One popular algorithm is the so-called best
response dynamic which allows the players to  take
turns choosing their best response (optimal strategy at a current time) to the state of the game in
the previous period.  If this process of iterative best responses converges, it is guaranteed to reach an NE of the game \cite{Saad}. Nevertheless, the convergence of best response dynamics is only guaranteed for certain classes of games, such as supermodular games \cite{Yildiz}. In order to use the best response algorithm,  we show that the proposed game is supermodular by reversing the order on one of the strategies.
Recall that we use a cost minimization framework. Thus, the supermodularity for our model is equivalent to the requirement that the utility function is twice continuously differentiable and
$\frac{\partial ^2 C_o(s,\nu)}{\partial s\, \partial \nu} \leq 0$ and $\frac{\partial ^2 C_r(s,\nu)}{\partial s\, \partial \nu} \leq 0$.

The second order derivatives are given as follows.
\begin{equation*}
\begin{split}
&\frac{\partial ^2 C_o(s,\nu)}{\partial s\, \partial \nu} = se^{-\nu s} + \alpha b e^{-\nu s} \geq 0.\\
&\frac{\partial ^2 C_r(s,\nu)}{\partial s\, \partial \nu} = (1-\alpha) b e^{-\nu s} \geq 0.
\end{split}
\end{equation*}
As a result, this leads to a \textit{submodular} game \cite{Yildiz}.

However, the proposed game can be converted into a \textit{supermodular} game when $\nu$ is ordered in the reverse order:
\begin{equation*}
\begin{split}
\frac{\partial ^2 C_o(s,\nu)}{\partial s\, \partial (-\nu)}  \leq0,~~
\frac{\partial ^2 C_r(s,\nu)}{\partial s\, \partial (-\nu)}  \leq 0.
\end{split}
\end{equation*}

Notice that we do not change neither the payoffs nor the structure of the game, we only alter
the ordering of the one player's strategy space. This approach only works in two-player games, and the submodular games with more than two players may exhibit dramatically different properties than the supermodular ones \cite{Yildiz}. Hence, for the proposed game, the decreasing best responses which means that each player's best strategy response function is decreasing in other player's strategy, will converge to the NE point. Moreover, for submodular games with finite strategies, the work in \cite{Topki} has proposed algorithms using the fictitious play to make the game dynamics converge to a pure equilibrium point.

\textbf{\textit{Remark 1}}: The value of $\alpha$ depends on a variety
of factors, such as the structure of the market, and the users' expectations.
At one extreme, if the RPS has a monopolistic  position
in the market with many competing BSs, $\alpha$
will be near $1$  and the BS has to sustain
the backlogged costs. At the other extreme, the RPS could be part of
a competitive RPS market, and, thus, a poor user service at the BS will mostly harm the RPS
as the cellular network can switch to another RPS. Then, $\alpha$ will be near $0$.
Clearly, our model allows to capture all such scenarios. In practice, the RPS' will eventually pay a certain cost for a delayed transmission due to the fact that delayed users might lead the network operator to either switch to another RPS or re-negotiate its contract with the RPS.

\textbf{\textit{Remark 2}}:
If the traffic load is served with a combination of a renewable
power source and the electric grid, we should design the adaptive power management scheme for the BS
to control the purchase of energy generated from the renewable source (with price $P_1$) and energy from the electric grid  (with price $P_2$).
The conventional energy has a stationary source. So that,
there is no need to account for the energy storage for using energy from the conventional source.
Let $\bar{\lambda}$ be the total arriving rate of connections.
Denote the arriving rate of connections served by the renewable energy as $\lambda$,
then the arriving rate of connections powered by the electric grid is $\bar{\lambda}-\lambda$.
Define $C_{ls}$ as the cost of the BS that schedules the load to be served by two kinds of energy sources. Based on the Nash equilibrium solution,  we get,
\begin{equation*}
\begin{split}
C_{ls} &= C_o^* + P_1\lambda + P_2(\bar{\lambda}-\lambda)\\
&=\frac{(\sqrt{1+\varphi}+f)\ln(1+\alpha b)}{ f\varphi} + P_1\lambda+P_2(\bar{\lambda}-\lambda),
\end{split}
\end{equation*}
where $\varphi = \frac{\mu_o}{\lambda} - 1$ and $C_o^*$ is the BS's cost at the NE. The optimal management discipline can be got by the first order derivation $\frac{\partial C_{ls}}{\partial \lambda} = 0$ or at the boundaries $\lambda=0$ and $\lambda=\bar{\lambda}$.

We assume that the energy price of the electric grid and the energy price of the RPS are fixed during one operation period. This is the case where the energy price is determined by the government or a contract between the user and the energy company \cite{Gprice}.
For example, a green pricing utility program in US department of energy shows that the green energy from a certain supplier has a fixed price \cite{Gprice}.
The fixed prices setting can help us to understand the interaction between the random green power generations and the dynamics of the energy consumption of mobile networks, especially considering
the QoS impact and decentralized energy allocation strategies.
If a dynamic pricing scheme is introduced between multiple users,
then the energy supply rate,  storage level, QoS cost and prices should be examined at the same time.
Currently, this challenging but interesting issue is still an open problem. Nonetheless, the developed model can serve as a building block for more elaborate, dynamic pricing mechanisms in future work.

\subsubsection{Centralized Optimal Solution}
Our performance benchmark for the decentralized system is the integrated/centralized
system, i.e., there is a single decision maker that simultaneously optimizes
the energy reservation level $s$ and the capacity variable $\nu$.
This may be the case when both the RPS and the BS belong to the same operator.
The total average expected normalized cost per unit time of the centralized system
can then be expressed as,
\begin{equation}\label{centralc}
\begin{split}
C(s,\nu) &= C_o(s,\nu)+ C_r(s,\nu)\\
&=s -\frac{1}{\nu} + (1+b) \frac{e^{-\nu s}}{\nu} +  c_s \frac{\nu+1}{\varphi-\nu}.
\end{split}
\end{equation}

We introduce an auxiliary variable, $\gamma = \ln(1+b)$.
The centralized solution is given in the following proposition.
\begin{proposition}
The optimal centralized solution is the unique solution to the first-order conditions, and is given by
\begin{equation}\label{optimal}
\bar{\nu} = \frac{\varphi\sqrt{c_s\gamma}}{\sqrt{c_s\gamma}+c_s\sqrt{\varphi+1}},~~
\bar{s} = \gamma\frac{\sqrt{c_s\gamma}+c_s\sqrt{\varphi+1}}{\varphi\sqrt{c_s\gamma}}.
\end{equation}
The resulting cost is
\begin{equation}\label{optimal-cost}
C(\bar{s},\bar{\nu} ) = \frac{1}{\varphi}\Big(c_s + \gamma  +2\sqrt{c_s\gamma(1+\varphi)}\Big).
\end{equation}
\end{proposition}

\begin{proof}
The function $C(s,\nu)$ defined in equation
(\ref{centralc}) is continuously differentiable and bounded below by $0$ in
$\mathbb{B} = \{(s,\nu)|s\geq 0,\nu> 0 \}$. Thus, a global minimum is either a local
interior minimum that satisfies the first-order conditions or an element
of the boundary of $\mathbb{B}$. From the first-order conditions, we get
\begin{equation}\label{svbar}
\left\{
\begin{aligned}
&\frac{\partial C(s,\nu)}{\partial s} = 0, \\
&\frac{\partial C(s,\nu)}{\partial \nu} = 0,
\end{aligned}\right.
\Rightarrow\left\{
\begin{aligned}
 s\nu &= \gamma,\\
 \nu &= \bar{\nu},\\
 s &= \bar{s}.
\end{aligned}\right.
\end{equation}

Substituting the above equations into equation (\ref{centralc}) yields equation (\ref{optimal-cost}).
The only interior point that is a candidate for the global minimum is $(\bar{s},\bar{\nu})$.
In addition, the Hessian of $C(s,\nu)$ at $(\bar{s},\bar{\nu})$ is given by,
\begin{equation*}
H(\bar{s},\bar{\nu}) = \left[
\begin{array}{cc}
\bar{\nu} & \bar{s}\\
\bar{s} & \frac{\gamma(\gamma+2)}{\bar{\nu}^3}+\frac{2c_s(\varphi+1)}{(\varphi-\bar{\nu})^3}
 \end{array}\right]
\end{equation*}

From (\ref{optimal}), we get $\bar{\nu} \leq \varphi$. Thus, $\bar{\nu}$ is a feasible value.
Also, $\gamma = \ln(1+b) > 0$, for $b>0$,~$c>0$.
Thus, the Hessian matrix is positive definite and
$(\bar{s},\bar{\nu})$ is the unique local minimum in the interior of $\mathbb{B}$.
Now, we consider the boundary values: $\lim\limits_{\nu\rightarrow 0}C(s,\nu)\rightarrow \infty~\text{for}~s\geq 0$,
$\lim\limits_{\nu\rightarrow \varphi}C(s,\nu)\rightarrow \infty~\text{for}~s\geq 0$,
$\lim\limits_{s\rightarrow \infty}C(s,\nu)\rightarrow \infty$. When $s=0$, the boundary value is
\begin{equation*}
C(0,\nu) = \frac{b}{\nu} + c_s\frac{\nu+1}{\varphi-\nu}.
\end{equation*}

Obviously, $C(0,\nu)$ is a convex function of $\nu$. The first order derivation of $C(0,\nu)$ is
\begin{equation}\label{b-nu}
\frac{\partial C(0,\nu)}{\partial \nu} = 0 \Rightarrow \nu =\frac{\sqrt{b}}{\sqrt{c_s(1+\varphi)}+\sqrt{b}}\varphi \leq \varphi.
\end{equation}
The solution in above equation is also a feasible value.
Then, substituting this solution into $C(0,\nu)$, we get
\begin{equation*}
C(0,\nu) \geq  \frac{c_s}{\varphi}\Big(1 + \frac{b}{c_s} +2\sqrt{\frac{b}{c_s}(1+\varphi)}\Big).
\end{equation*}
Since $b \geq \ln(1+b)$, we derive $C(0,\nu) \geq C(\bar{s},\bar{\nu})$.
Hence, $(\bar{s},\bar{\nu})$  is the unique global minimum for $C(s,\nu)$.
\end{proof}

As expected, the optimal power supply rate for the BS decreases
with the supply cost $c_s$ and increases with the
backlogged cost $b$. Similarly, because
supply rate and power reservation provide alternative means
to avoid backlogged connection requests, the optimal base-reservation level
increases with the supply cost and with the normalized
backlogged cost. Finally, we can observe that the cost-split factor $\alpha$ plays no
role in this central optimization, because the utility transfer between the BS and the RPS does not
affect the centralized cost.

Note that the centralized optimal solution is based on the assumption that the BS and the RPS belong to the same operator. However,  the BS and the RPS may not have a common benefit such as they belong to different operators \cite{Ansari}.
In this sense, each one will seek to selfishly optimize its individual cost and, in this case,
the NE can be considered as the solution in such a decentralized system.
It is clear that the decentralized operation is less efficient in terms of the total
system profit. In next subsection, we will propose a contract-based approach which can be used to improve the decentralized operation.

\subsubsection{Comparison of Solutions and Coordination via Contracts}
According to (\ref{nash-s}) and (\ref{svbar}) , the first-order conditions
are $s\nu = \ln(1+b)$ in the centralized solution and
$s\nu = \ln(1+\alpha b)$ in the Nash solution. Hence, the two solutions
are not equal when $\alpha < 1$. And, as discussed earlier, the NE
in the $\alpha = 1$ case is an unstable system.

The magnitude of the inefficiency of a NE
is typically quantified by comparing the costs
under the centralized and Nash solutions. We compute the competition
penalty~\footnote{Note that this concept is different from the popular concept of price of anarchy\cite{Saad} which is defined as the ratio between the ``worst equilibrium'' and the optimal ``centralized'' solution.}, which is defined as the percentage increase in \textit{variable cost} of the NE over
the centralized solution,
\begin{equation}
\mathcal{P}_{\alpha,c_s} = \frac{C_r^* + C_o^*-C(\bar{s},\bar{\nu})}{C(\bar{s},\bar{\nu})}= \frac{C_r^* + C_o^*}{C(\bar{s},\bar{\nu})} - 1.
\end{equation}

It is clear that, in general, the decentralized operation is less efficient in terms of the total
system profit. A coordination mechanism based on static transfer payments can be used to improve the
decentralized operation. In particular, the transfer payment  modifies the cost functions
in (\ref{Cos2}) and (\ref{Cr2}) for the BS and RPS, respectively,~to
\begin{equation*}
\begin{split}
&\tilde{C}_o(s,\nu) = C_o(s,\nu)-\varepsilon'(s,\nu),\\
&\tilde{C}_r(s,\nu) = C_r(s,\nu)+\varepsilon'(s,\nu).
\end{split}
\end{equation*}

The choice of $\varepsilon'$ that coordinates the system is not unique.
One possibility is to define the transfer in such a way that the modified cost functions replicate
a cost-sharing situation. That is, we can set $\varepsilon$ such
that $\tilde{C}_r(s,\nu) =  \varepsilon C$, $\tilde{C}_o(s,\nu) =  (1-\varepsilon) C$
where $\varepsilon \in [0,1]$ is a splitting factor and $C$ is the
centralized cost function defined in equation (\ref{centralc}).
Such a cost sharing forces the transfer payment to satisfy,
\begin{equation*}
\varepsilon'(s,\nu) = \varepsilon C_o(s,\nu) -(1-\varepsilon)C_r(s,\nu).
\end{equation*}

With such a payment transfer, the RPS and the BS have aligned objectives and
the centralized solution $(\bar{v},\bar{s})$ is the unique NE.
The BS and RPS must be better off under the NE with the transfer
payments than under the NE without
the transfer payments, for example,
\begin{equation*}
C_r^*  \geq \varepsilon C , ~C_o^* \geq (1-\varepsilon)C.
\end{equation*}
After algebraic manipulation, we get
\begin{equation*}
\varepsilon \in \left[\frac{C_r^*}{C}- \mathcal{P}_{\alpha,c_s},~ \frac{C_r^*}{C} \right] \cap \big[0,1\big].
\end{equation*}

From a practical perspective, a direct cost-sharing negotiation based on total cost
probably dominates the transfer-payment approach
because it does not require any special accounting of
cost (in terms of supply and backlogged costs).

\textit{\textbf{Remark 3}}. A cooperative game may yield better results than the \textit{pure} non-cooperative result. In cooperative games, there are various solution concepts, i.e., Nash bargaining solution, core, Shapley value and so on. In order to enable the cooperative formulation, we should define the profit/energy cost saving that each
single player could generate.  For example, we can define the
single player's payoff as the utility it obtains in the non-cooperative setting.
Then, the cooperative solution (e.g., via a Shapley value) will tell us how important each player is to the overall cooperation, and what payoff he or she can reasonably expect.
Indeed, the non-cooperative game with contract coordination also generates a system optimal solution, and thus the coordination can be viewed as transferable utilities in a cooperative setting.
In addition, cooperative games which mainly study the interactions
among coalitions of players a suitable framework for handling
dense networks and can be studied in future work.

\section{Capacity Allocation Game with Multiple BSs}
In this section, we consider a setting in which a single \textit{monopolistic} RPS sells energy to a set of $\mathcal{N} = \{1...,N\}, N\geq 2$ BSs.
The RPS can produce no more than $\mu_0$ units during the period.
We consider that UEs of each BS cannot switch to another BS. This is the case for which BSs are owned by different operators~\footnote{If the BSs are under a price competition in spectrum accessing market.
The energy allocation mechanism and the spectrum access scheme will affect each other. As a result,
the properties of allocation mechanisms would be significantly different, and are not currently investigated here.}. BSs are considered to have different arrival rates due to a variety of reasons including geographic locations, operator promotion plans, and pricing strategies.
The RPS has a monopoly over the energy market and is thus considered to not bear backlogged cost and the load increasing cost, i.e., $\alpha = 1$, $c_s = 0$ in (\ref{Cr2}).
We will investigate how RPS allocates the limited energy capacity to BSs, and how induce the BSs to truthfully report their optimal demand.
Toward this end, we will propose a noncooperative capacity allocation game and address its dominant
equilibrium.

First, we  examine the cost function of a BS in considered setting.
Let $\mathbf{\Gamma}=\{\mu_1,...,\mu_N\}$ be the energy supply rate vector of the RPS, and
$s_i$ be the energy reservation strategy of BS $i$.
Then, the normalized energy supply rate of BS $i$ is denoted as $\nu_i$ and the
backlogged cost of BS $i$ is $b_i$.
As discussed in Section 4.2.2, as $\alpha\rightarrow 1$, the RPS will set
its energy supply rate $\mu\rightarrow 0 $ for a BS.
In order to prevent this outcome in the monopolistic market,
we set that BS $i$ is charged an incentive price $p$ to the RPS for the supply rate $\mu_i$.
Based on the M/M/1 make-to-stock queuing analysis,
the steady-state expected normalized cost per unit time for BS $i$ is,
\begin{equation*}\label{Cgamecost}
\begin{split}
&\bar{C}_i(\mu_i,s_i,\lambda_i)\\
&= p\mu_i+P_1\lambda_i+P_2(\bar{\lambda}_i-\lambda_i) + C_o^i(s_i,\nu_i)|_{\alpha=1},\\
&= p\mu_i+P_1\lambda_i+P_2(\bar{\lambda}_i-\lambda_i)+s_i - \frac{1-(b_i+1)e^{-\nu_i s_i}}{\nu_i},
\end{split}
\end{equation*}
where $\lambda_i$ is the arrival rate of connections served with the renewable energy of BS $i$,
and $\bar{\lambda}_i$ is the total arriving rate of connections of BS $i$.
$C_o^i(s_i,\nu_i)|_{\alpha=1}$ is backlogged cost and store self-discharge cost defined in (\ref{Cos2}).

\begin{proposition}
The global minimum of $\bar{C}_i(\mu_i,s_i,\lambda_i)$ is the solution that satisfies,
\begin{equation*}
\hat{\mu}_i=\sqrt{\frac{\lambda_i\ln(1+b_i)}{p}}+\lambda_i,~\hat{s}_i=\sqrt{p\lambda_i\ln(1+b_i)}.
\end{equation*}
The resulting cost is
\begin{equation*}
\bar{C}_i^{\text{min}}(\hat{\mu}_i,\hat{s}_i,\lambda_i)=2\sqrt{p\lambda_i\ln(1+b_i)}+(p+P_1)\lambda_i+P_2(\bar{\lambda}_i-\lambda_i).
\end{equation*}
\end{proposition}
\begin{proof}
We introduce the following auxiliary function,
\begin{equation*}
C'_i(\mu_i,s_i) =p\mu_i+s_i - \frac{1-(b_i+1)e^{-\nu_i s_i}}{\nu_i}.
\end{equation*}
The solution to the first-order conditions of $C'_i(\mu_i,s_i)$
\begin{equation*}
\begin{split}
\frac{\partial C'_i(\mu_i,s_i,\lambda_i)}{\partial s_i} = 0,\Rightarrow&s_i\nu_i=\ln(1+b_i),\\
\frac{\partial C'_i(\mu_i,s_i,\lambda_i)}{\partial \mu} =0,\Rightarrow&\Big[-(b_i+1)
(s_i\nu_i+1)\frac{e^{-s_i\nu_i}}{\nu_i^2}+\frac{1}{\nu_i^2}\Big]\frac{1}{\lambda_i}\\
&+p=0,\\
\end{split}
\end{equation*}
is given by $\hat{\nu}_i=\frac{\hat{\mu}_i-\lambda_i}{\lambda_i}=\sqrt{\frac{\ln(1+b_i)}{(p\lambda_i)}}$, $\hat{s}_i=\sqrt{p\lambda_i\ln(1+b_i)}$. Using a technique similar to Theorem 1, we can prove that the Hessian matrix of $C'_i(s_i,\nu_i)$ at $(\hat{s}_i,\hat{\nu}_i)$
is positive definite and the boundary values can be excluded.
Thus, the global minimum of $C'_i(\mu_i,s_i)$ is the unique solution
to the first-order conditions. And, the minimum cost is $2\sqrt{p\lambda_i\ln(1+b_i)}+p\lambda_i$.
\end{proof}

The individual BSs have private information regarding the
optimal energy demand, and will competitively submit their orders to the RPS.
The RPS allocate the energy supply rates to the BSs according to an allocation mechanism.
Define $\mathcal{A}=\{\boldsymbol{a}\in\mathbf{R}^N: \boldsymbol{a}_i\geq 0~\text{and}~ \sum_{i=1}^N\boldsymbol{a}_i\leq \mu_0\}$ as a set of allocations. We call each
$\boldsymbol{a}\in \mathcal{A}$ as a feasible allocation.
Let $\mathbf{m}$ be the renewable energy supply rate vector that BSs order with each element $m_i$ being the supply rate ordered by BS $i \in \mathcal{N}$. Let $\mathbf{m}_{-i}$ be the vector of
other BSs' orders. Then, we have the definition of an allocation mechanism.
\begin{definition}
An \emph{allocation mechanism} is a function $g$ that assigns a feasible allocation to each vector of
orders, $g(\mathbf{m})\in\mathcal{A}$. Then, $g_i(\mathbf{m})$ is BS $i$'s energy supply rate.
\end{definition}

The RPS can never allocate to a BS more than the BS orders, i.e., $g_i(\mathbf{m})\leq m_i$.
The RPS should choose an allocation mechanism $g$  and then broadcast the
mechanism to all BSs. All BSs will simultaneously submit their energy demand orders to the RPS according to the allocation mechanism $g$. Based on both the connection demand and
the noticed allocation mechanism, each BS $i$ determines its ordered energy supply rate $m_i$.
Due to the limited capacity of the RPS, BSs will compete with each other
in order to get their individual favorable energy supply. For instance, BSs can report larger energy demands rather than optimal values to get more allocations.
Such a competition leads to the formulation of a \textit{capacity allocation game} in which players are BSs, and the strategy of each player is its ordered renewable energy supply rate $m_i$.  The cost function of BS $i$ are given by $\bar{C}_i(\mu_i,s_i,\lambda_i)$ with $\mu_i=g_i(\mathbf{m})$. Then, we state the following equilibrium definition.

\begin{definition}
Assume that all BSs truthfully reporting their optimal demands, $\mathbf{m}^*$.
Then, reporting $\mathbf{m}^*$ is a \emph{dominant equilibrium} of the capacity allocation game under mechanism $g$, if and only if $\forall~\mathbf{m}$,
\begin{equation*}\label{de}
\bar{C}_i(g_i(m_i^*,\mathbf{m}_{-i}),s_i,\lambda_i) \leq \bar{C}_i(g_i(m_i,\mathbf{m}_{-i}),s_i,\lambda_i),~\forall~i\in\mathcal{N}.
\end{equation*}
\end{definition}

In a dominant equilibrium, each BS has an energy order that minimizes its cost regardless of the energy orders of the other BSs. A dominant equilibrium is a stronger notion of equilibrium than the Nash equilibrium (NE) in which each player is assumed to know the equilibrium strategies of the other players, and no player has anything to gain by changing only its own strategy unilaterally \cite{Osborne}. In this regard, the NE definition of the game can be obtained by replacing $\mathbf{m}_{-i}$ with $\mathbf{m}_{-i}^*$. We are particularly interested in RPS's allocation mechanism under which the BSs report their optimal energy supply demands in a dominant equilibrium.

More specifically, if BSs ordered exactly their needs, the RPS could determine how much capacity it needs to allocate. Conversely, manipulation may generate undesirable consequences for the system, i.e., preventing the RPS from determining which BS is in reality needing the most energy. Some BSs with high expected demands may receive too little and others with low expected demand may receive too much. At the end, the system ends up serving all BSs poorly. Note that, the allocation mechanism design implies that all BSs will observe the same price for green energy. Such a model allows to capture scenarios in which having an auction mechanisms may be too complicate or in which there is a single, unified price imposed by a utility company or the government.

In what follows, we will study the mechanism design properties, and subsequently we will investigate the equilibrium of the capacity allocation game.
The two main challenges of mechanism design are stability and efficiency.
Two typical criteria are generally applied to mechanism design:
incentive compatibility and optimality defined as follows:
\begin{definition}
An allocation mechanism $g$ is said to be \emph{incentive compatible} (IC) or truth-inducing if the case in which all BSs place their orders truthfully at their optimal profits constitutes an dominant equilibrium of $g$.
\end{definition}
\begin{definition}
An allocation mechanism $g$ is a \emph{Pareto allocation mechanism} if it maximizes
the sum of BSs' profits assuming all BSs truthfully submit their optimal orders.
\end{definition}

In general, Pareto optimality is a state of allocation of resources in which it is impossible to make any one individual better off without making at least one individual worse off.
Here, we use the definition of Pareto allocation mechanism following \cite{Cachon}.
Note, the social welfare maximization implies Pareto optimality, whereas the versa does not hold.
One of the most popular allocation mechanism is the so-called proportional allocation, for which:
$g_i(\mathbf{m})=\min\{m_i,\mu_0\frac{m_i}{\sum_{j=1}^Nm_j}\}.$

In \cite{Cachon}, it is shown that for a kind of utility
functions, the proportional allocation is actually a Pareto allocation mechanism.
Next, we show that the proportional allocation mechanism is not a Pareto
mechanism when the total renewable power production capacity  is less than total
BSs' orders.
\begin{lemma}
If BSs have different connection arrival rates and order truthfully their optimal energy supply rates $\mathbf{m}$ ($\mu_0<\sum_i^N{m_i}$), the proportional allocation is not a Pareto mechanism.
\end{lemma}
\begin{proof}
Based on Proposition 3, the optimization problem of total BSs' profits
can be formulated as follows:
\begin{equation*}\vspace{-0.5cm}
\begin{split}
&\min_{\mathbf{\Gamma}\in\mathcal{A}}\sum_{i=1}^{N}\Big[2\sqrt{p\lambda_i\ln(1+b_i)}
+(p+P_1-P_2)\lambda_i+P_2\bar{\lambda}_i\Big].\\
\vspace{-0.6cm}
&\text{s.t}.~~~~
\end{split}
\end{equation*}\vspace{-0.5cm}
\begin{equation}
\mu_i = \sqrt{\frac{\lambda_i\ln(1+b_i)}{p}}+\lambda_i,
\end{equation}\vspace{-0.5cm}
\begin{equation}
\sum_{i=1}^N \mu_i\leq \mu_0,~\lambda_i\leq \bar{\lambda}_i,~\mu_i\geq 0.~i\in\mathcal{N}.
\end{equation}\vspace{-0.5cm}

The Larangian function of the original problem is,
\begin{equation*}
\begin{split}
&\mathcal{L}(\epsilon,\eta,\lambda)\\
&=\sum_{i=1}^{N}\Big[2\sqrt{p\lambda_i\ln(1+b_i)}+(p+P_1)\lambda_i+P_2(\bar{\lambda}_i-\lambda_i)\Big]\\
&+ \epsilon(\sum_{i=1}^N\mu_i-\mu_0)+ \sum_{i=1}^{N}\Big[\eta_i(\lambda_i-\bar{\lambda}_i)\Big].
\end{split}
\end{equation*}
where $\epsilon\geq0$ and $\eta_i\geq 0$, $i=1,2,...,N$.

The optimal solution $\lambda_1^*,...,\lambda_N^*$ satisfies
the necessary Karush-Kuhn-Tucker (KKT) conditions which are,
\begin{equation}
\frac{\partial}{\partial \lambda_i^*}\mathcal{L}(\epsilon,\eta,\lambda^*)=0,~i=1,...,N,
\end{equation}\vspace{-0.3cm}
\begin{equation}
\epsilon(\sum_{i=1}^N\mu_i^*-\mu_0)= 0, \sum_{i=1}^N\mu_i^*-\mu_0 \leq 0,
\end{equation}\vspace{-0.3cm}
\begin{equation}
\eta_i(\lambda_i^*-\bar{\lambda}_i) = 0,~\lambda_i^*-\bar{\lambda}_i \leq 0,~\lambda_i^*\geq 0,~i\in\mathcal{N}.
\end{equation}
where $\mu_i^*$ is defined by (18) with $\lambda_i=\lambda_i^*$.
By further developing (20), we obtain,
\begin{equation}
\sqrt{\lambda_i^*}=\frac{(2p+\epsilon)\sqrt{\ln(1+b_i)}}{2\sqrt{p}(P_2-p-P_1-\epsilon-\eta_i)}.
\end{equation}

Assume that the proportional allocation mechanism is a Pareto mechanism, we then have
$\mu_i^* = g_i(\mathbf{m})$ = $\mu_0\frac{m_i}{\sum_{j=1}^Nm_j} < \bar{\lambda}_i$.
Then, from (22), we get $\eta_i=0~\forall~i \in \mathcal{N}$.
Thus, from equation (23), we derive
$\lambda_1^*= \lambda_2^*=...=\lambda_N^*$ which contradicts with the proportional allocation.
This completes the proof.
\end{proof}

If $P_2\leq P_1+p$, all BSs will order a zero renewable energy supply rate, thus we just consider
the case in which $P_2> P_1+p$.
The solution to  $\bar{C}_i^{\text{min}}(s_i,\nu_i,\lambda_i)$$=P_2\bar{\lambda_i}$
is $\hat{\lambda}_i=\frac{4c\ln(1+b_i)}{(P_2-P_1-p)^2}$.
The objective function of the optimization problem is a concave function as shown in
Fig. \ref{optimality}.
Thus, the global minimum is achieved at the boundary.
In this respect, when all BSs have an identical $b_i,i\in\mathcal{N}$, a simple \textit{Pareto mechanism} discipline can be designed as follows. Assuming that all BSs truthfully report their demand, a larger demand will be satisfied with a certain priority. If a BS receives an allocation that is less $\hat{\lambda}_i$,  the allocation will be rejected. The RPS then adjusts the allocation for this BS to zero, and the BS will wait for the next allocation period to announce its demand.
\begin{figure}[t]
\centering
\includegraphics[width=3.0in]{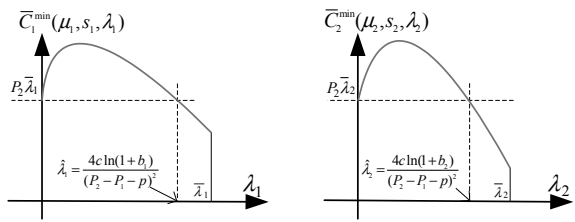}
\vspace{-0.2cm}
\caption{An example to show the concavity of the objective function.}\label{optimality}
\end{figure}

When the demand exceeds the capacity size, the BSs may inflate orders so as to be
allocated more than they need. Hence, by increasing the order quantity, each BS
is able to decrease the allocation quantity for competitors.
Next, we investigate how to design a mechanism which can induce the BSs truthfully to report their energy demand.  We arrange the BS in decreasing order of their order quantities, i.e., $\{m_1\geq m_2\geq...\geq m_N\}$, $\hat{n}$ be the largest integer less than or equal to $N$ such that
$g_{\hat{n}}(\mathbf{m},\hat{n})\leq m_{\hat{n}}$.
Then, $\frac{1}{\hat{n}} (\mu_0-\sum_{\hat{n}+1}^Nm_j)$ means that BSs other than those whose indices are larger than $\hat{n}$, will get an \textit{uniform} allocation. Notice $\hat{n}$ is inherently determined by the allocation mechanism, and is not specified by the RPS.

In this regard, we propose the adaptive uniform allocation mechanism shown in Algorithm~1. Under such a mechanism, initially, the BSs with orders less than a threshold index $\hat{n}$ receive the same energy supply rate as their orders, and the rest of the BSs receive $m_i$.  Then, the remainder of the capacity is divided by the number of BSs indexed greater than $\hat{n}$. If a BS receives an allocation in less than $\hat{\lambda}_i$, it will reject the allocation. The RPS adjusts the allocation for this BS to zero. Thus, it also can be viewed as a take-or-leave choice for each BS. Note that the uniform allocation mechanism always favors small BSs.
\begin{algorithm}[t]
\caption{Proposed adaptive uniform allocation mechanism.}
\label{alg:Framwork}
\begin{algorithmic}[1]
\REQUIRE~~\\
\STATE Arrange the BS in decreasing order of their order quantities, i.e., $\{m_1\geq m_2\geq...\geq m_N\}$. $\hat{n}$ is the \textit{largest} integer less than or equal to $N$ such that
$g_{\hat{n}}(\mathbf{m},\hat{n})\leq m_{\hat{n}}$ where,
\begin{equation}
g_i(\mathbf{m},\hat{n}) = \left\{
\begin{aligned}
&\frac{1}{\hat{n}} \Big(\mu_0-\sum_{\hat{n}+1}^Nm_j\Big),i\leq \hat{n},\\
&m_i,~~i>\hat{n},
\end{aligned}
\right.
\end{equation}
\ENSURE ~~\\
\IF{$0<g_i(\mathbf{m})\leq \hat{\mu}_i = \sqrt{\hat{\lambda}_i\frac{\ln(1+b)}{p}}+\hat{\lambda}_i$, $i\in \mathcal{N}$}
\STATE The BS $i$ sends a message indexed $i$ to the RPS;
\ENDIF
\REQUIRE~~\\
\IF{the RPS receives a message $i$}
\STATE $g_i(\mathbf{m}) = 0$;
\ENDIF
\renewcommand{\algorithmicensure}{\text{Output:}}
\ENSURE  $g(\mathbf{m})$, which is an energy supply rate allocation for BSs.
\end{algorithmic}
\end{algorithm}

\begin{theorem}
If the RPS uses the proposed allocation mechanism in Algorithm 1, it will minimize the cost for each BS
to report its optimal energy demand, regardless of the energy orders of the other BSs, thereby reaching a dominant equilibrium.
Moreover, the proposed mechanism is truth-inducing.
\end{theorem}
\begin{proof}
Assume that all BSs truthfully report their optimal demands, $\mathbf{m}^*$. Then, we
should prove that reporting $\mathbf{m}^*$ is a dominant equilibrium under the adaptive uniformly allocation mechanism. We observe from Algorithm 1 that, by ordering more than $m_i^*$, there will be two cases for BS $i$.

The first is that BS $i$ received an allocation that is less than $m_i^*$ under truthful reporting. In this case, inflating the order cannot result in more allocation, since the BSs whose indices are smaller than $\hat{n}$ will be uniformly allocated an energy supply rate.
The second case is the one in which a certain BS $i$ received $m_i^*$ under the truthfully reporting.
In this situation, if there are some other BSs to inflate their orders, BS $i$ still receives $m_i^*$
because the mechanism always favors a BS with a smaller demand. Also, if some other BSs
reduce their orders, there will be enough available capacity for BS $i$ to get its optimal energy supply rate demand $m_i^*$ under such a mechanism.

Similarly, if a BS $i$ received a zero allocation, ordering less $m_i^*$ still results in a zero allocation. For a BS $i$ which received an allocation in less than $m_i^*$, ordering less $m_i^*$ cannot increase its allocation.
For the remaining cases in which a BS $i$ received $m_i^*$, ordering less than $m_i^*$ will reduce the its allocation. Thus, using the proposed mechanism, all BSs will truthfully report their optimal energy demands, which also constitutes a dominant equilibrium.
\end{proof}

It is easy to see that adaptive uniformly allocation is not necessarily efficient. But,
by choosing the IC mechanism, the RPS can acquire truthful energy demand information of BSs.
This could lead the RPS's secure decision-makings on capacity planning and sales planning.
\vspace{-0.35em}
\section{Simulation Results and Analysis}
In this section, we study how the RPS allocate its energy supply rate
with a limited capacity and how the BS
optimize their renewable energy storage based on the predictable traffic
condition. Further, the energy supply allocation with multiple BSs is also shown.
\vspace{-0.35em}
\subsection{Parameters setting}
For our simulations, we use the following parameters.
The time unit is set to one minute,  while the dynamic power consumption
coefficient is $24$ W (Joule/s) per connection, which is similar to that in \cite{Niyato}\cite{Arnold}.
Without loss of generality, the time unit $t_d\cdot\text{TTI}$ is set to be $1$ second.
Thus, the energy unit is $24$ Joule.
According to a report from an independent carbon footprints research group,
the average national electrical price is $12$ cents/kWh in the US
\cite{Eprice}. Other countries such as
Denmark, Germany and Spain may have more expensive electricity prices that can exceed $30$ cents/kWh.
For the green energy price,  a green pricing utility program in US department of energy, shows that
prices of different green power suppliers vary from about $0.5$ cents/kWh to $4.5$ cents/kWh
\footnote{The price of green energy is often expressed as a price ``premium'' above the
price of conventional power  in order to use the blended green and conventional
power\cite{Gprice,Guide}.} \cite{Gprice}. We can observe that the green power could be much cheaper than the conventional power.

We use cents as the power price unit. For example, the renewable energy has a low price
(e.g., $2.5$ cents per kWh), and the electric energy has a high price (e.g., $25$ cents per kWh).
These values fall within the range of the prices discussed in \cite{Eprice} and \cite{Gprice}.
Then, the energy cost of a connection served by the electric source will be 0.01 cents per minute,
and the energy cost of a connection served by the renewable source is 0.001 cents per minute.
The energy reservation cost factor $c$ is set to 0.001.
Thus, the reservation cost with $I_s$ units average energy storage is $0.001I_s$ cents per minute.
The units of the backlogged cost $b$, and the supply cost of the RPS $c_s$ are also set to cents per minute. In simulations, the cost accounts for 1000 minutes (16.7 hours), i.e,
we can set $P_1 = 1$, $P_2=10$ and $b=10$ cents per 1000 minutes.

\begin{figure}[t]
\centering
\includegraphics[width=80mm]{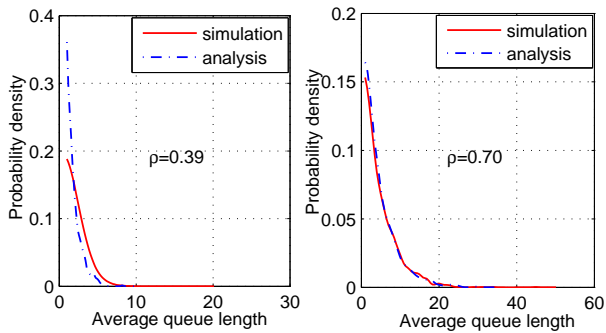}
\vspace{-0.3cm}
\caption{The PDF in different traffic conditions.}\label{queuelength}
\end{figure}
\begin{figure}[t]
\centering
\includegraphics[width=65mm]{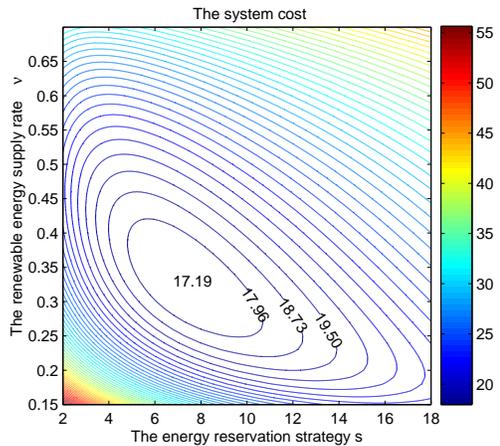}
\vspace{-0.5cm}
\caption{The central system cost.}\label{central}
\end{figure}
\begin{figure*}[t] \centering
\subfigure[The minimum system cost.]{ \label{fig:subfig:a}
\includegraphics[width=2.2in]{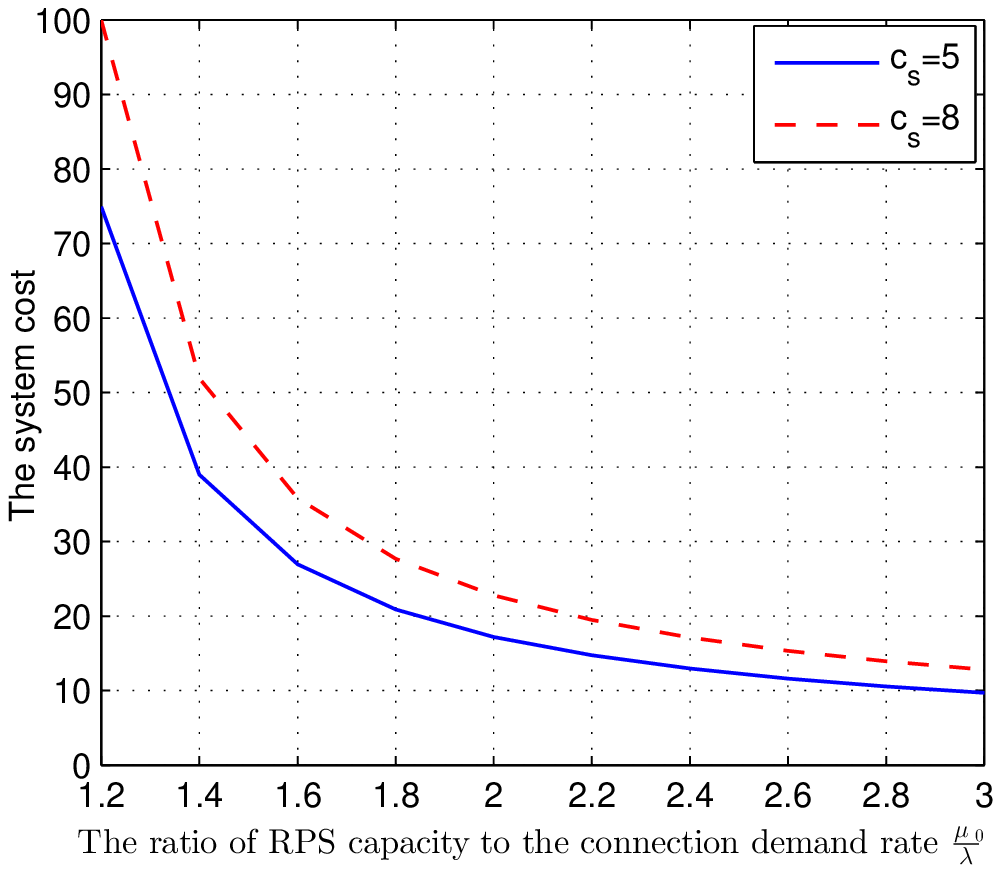}}\hspace{-0.1in}
\subfigure[Optimal energy supply rate.]{ \label{fig:subfig:b}
\includegraphics[width=2.2in]{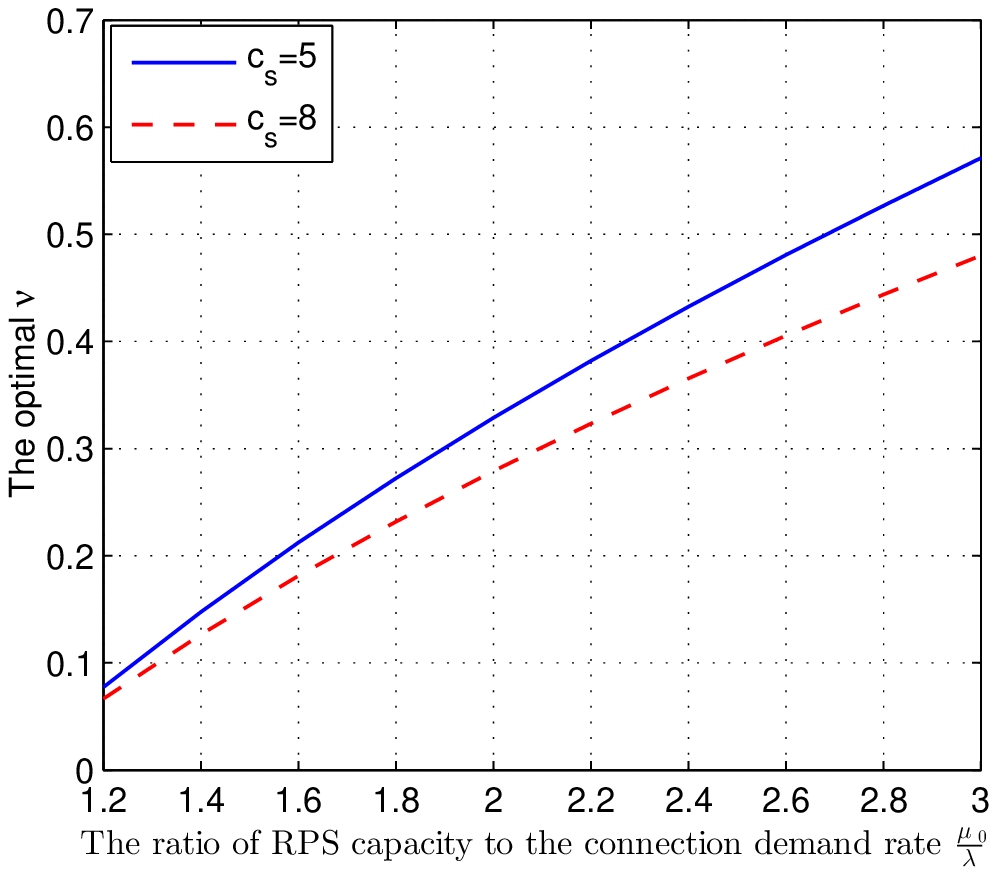}}\hspace{-0.1in}
\subfigure[Optimal energy reservation.]{ \label{fig:subfig:b}
\includegraphics[width=2.2in]{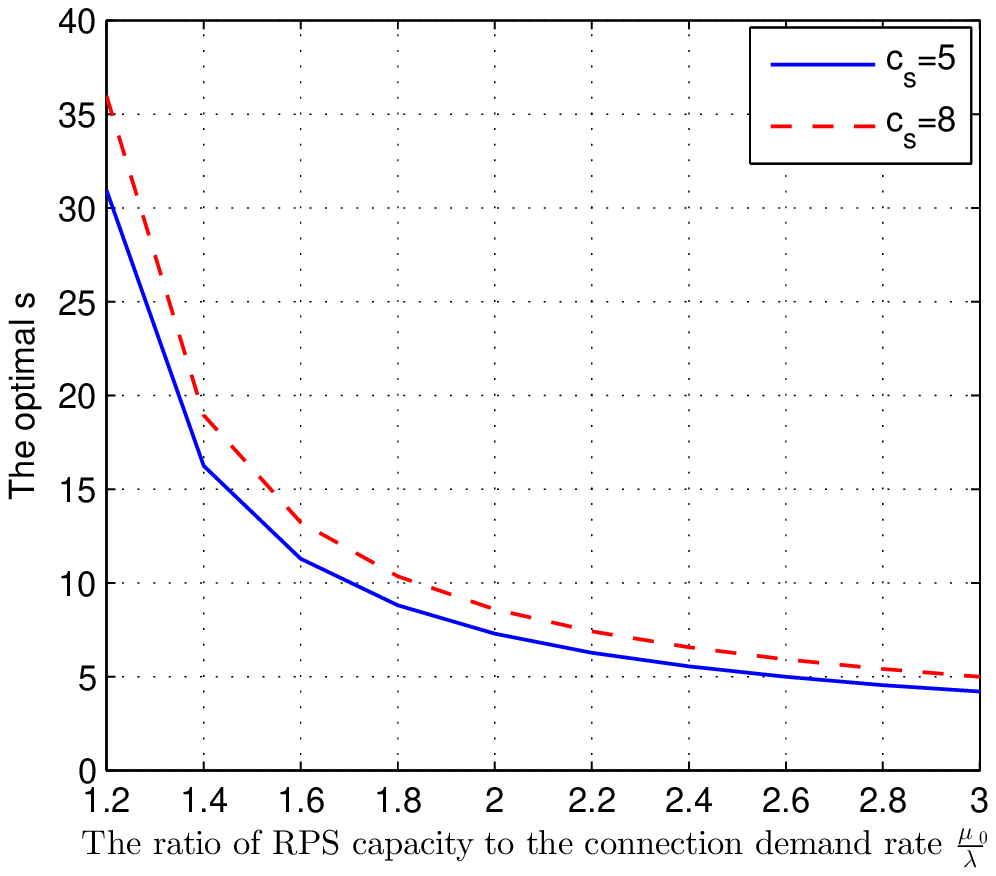}} \vspace{-0.3cm}
\caption{The minimum system cost,
 optimal energy supply rate and energy reservation with different parameters.} \label{central2}
\end{figure*}

\begin{figure*} \centering
\subfigure[The cost of BS with $\nu=\nu^*$]{ \label{fig:subfig:a}
\includegraphics[width=2.2in]{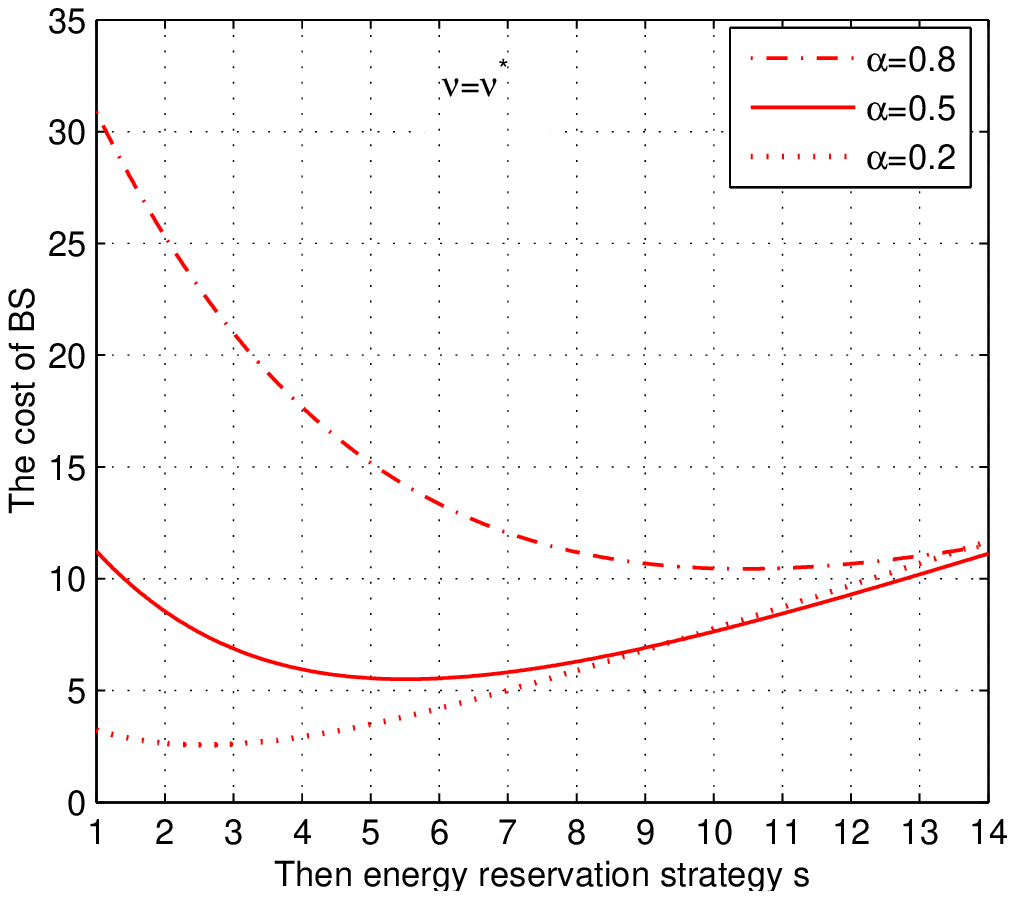}}\hspace{-0.1in}
\subfigure[The cost of RPS with $s=s^*$.]{ \label{fig:subfig:b}
\includegraphics[width=2.2in]{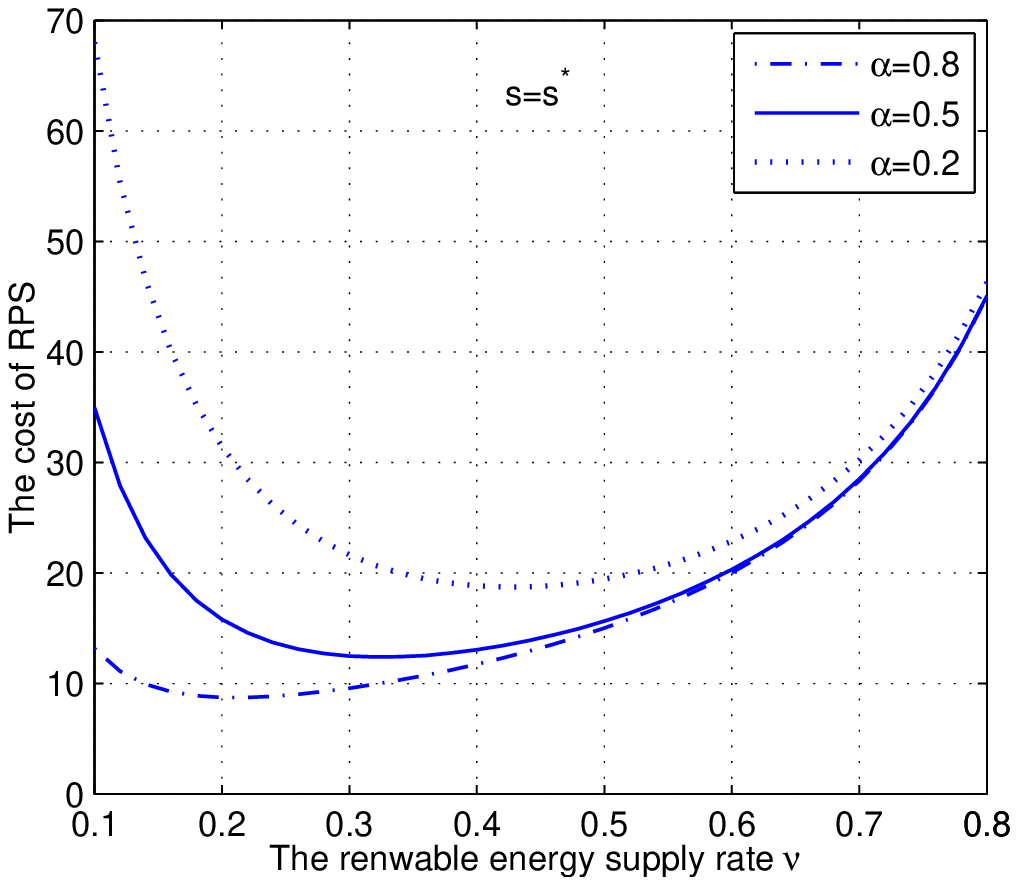}}\hspace{-0.1in}
\subfigure[The response of RPS with different load factor at $s=s^*$.]{ \label{fig:subfig:b}
\includegraphics[width=2.2in]{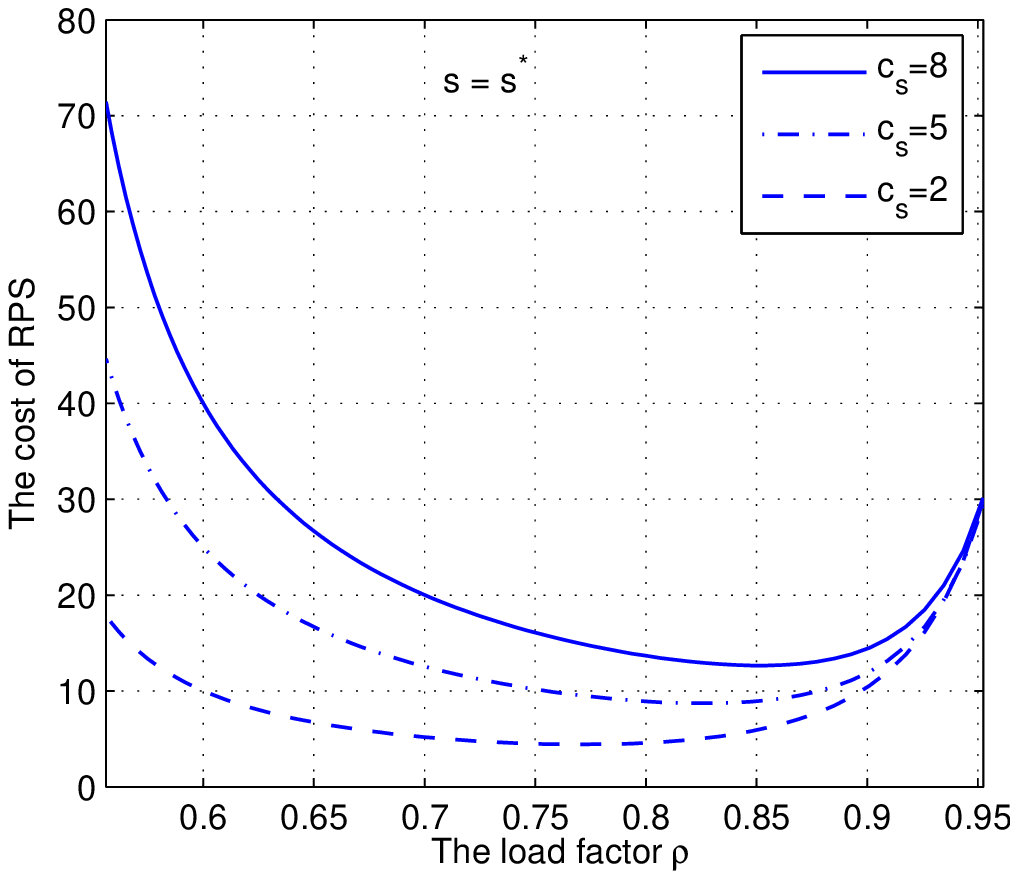}}\hspace{-0.1in}\vspace{-0.3cm}
\caption{Illustration of the NE.} \label{figNE}
\end{figure*}
\vspace{-0.35em}
\subsection{The heavy traffic approximation}
In Table \ref{tbl-len} and Fig. \ref{queuelength}, we illustrate that the continuous-state approximation can be justified by a heavy traffic approximation with an incorporation of general inter-arrival time and service time distributions. For instance, we assume that the traffic of the BS alternates between a high state and a low state, where the traffic is in ``high'' with the connections' inter-arrival time exponentially distributed with rate $3.5$ (e.g., connections/second) and in the ``low'' state with such a rate $2.3$. The traffic model can be viewed as a 2-type hyper-exponential distribution, i.e., there are two different mobile services with different request rates. We consider that the two type services arrive with an equal probability. Thus, the average connections' inter-arrival time is $\frac{1}{2}\frac{1}{2.3}+\frac{1}{2}\frac{1}{3.5}=0.36$.
We consider that the renewable energy supply rate follows a normal distribution with the mean $\mu$
\footnote{The normal distribution is often used to model random noise. Here, we use such a distribution for our simulations. But, our analysis also applies to general distributions under heavy traffic situations.}.
Table~\ref{tbl-len} shows the average queue length where $\rho=\frac{1}{0.36\mu}$.
And, the probability density function (PDF) of the queue length is plotted in
Fig.~\ref{queuelength}. We can observe that
the continuous-state approximation is quite accurate when the average queue length is equal or larger than~$2$.
\begin{table}[h]
\vspace*{-0.7em}
\caption{The average queue length under different traffic conditions.}\label{tbl-len}
\vspace*{-1.5em}
\begin{center}
\begin{tabular}{|l|l|l|l|l|}\hline
  ~&$\rho=0.39$ &$\rho=0.70$&$\rho=0.80$&$\rho=0.93$\\
  \hline
  Analysis & 1.0& 2.3& 4.3&14.6\\
  \hline
  Simulation & 0.7& 2.4& 4.2&14.2\\
  \hline
\end{tabular}
\end{center}
\end{table}
\vspace{-0.35em}
\subsection{Analysis for the single BS scenario}
Fig.~\ref{central} verifies the results pertaining to the centralized global minimum of the system cost where $\bar{\nu} = 0.33$, $\bar{s} = 7.29$, and $C(\bar{\nu},\bar{s}) = 17.19$ for
$b=10,~c_s=5,~\varphi=1$. Recall that $\nu$ is the normalized energy supply rate, i.e.,
$\nu=\frac{\mu-\lambda}{\lambda}$. Thus, the load factor is $\rho = \lambda/\mu = 0.75$ in the
global minimum state. The system consists of the energy reservation cost and the renewable energy
supply cost. Each of them depends on both the values of $s$ and $\nu$. We can observe that for a certain system cost (e.g., 17.96), $s\cdot\nu$ will be a fixed value. This means that increasing $s$ has the same effect as reducing $\nu$ to achieve such a  system cost. Thus, it is easy to imagine that there will be
a competition between adjusting $s$ and $\nu$ in the decentralized system.

The system cost and the Nash solution are only related to the ratio of the capacity of
the RPS to the connection demand rate, i.e., $\frac{\mu_0}{\lambda}$.
The output of the RPS depends on the weather condition, the size of power generator i.e., the solar panel, and the time on one day. For example, in \cite{Hoff2}, it is shown  that the 10-second irradiance is about $500$ W/m$^2$ at a measure point in Canada over a clear sky at 12:00, and this value reduces to
about $200$ W/m$^2$ at 10:00.  The efficiency of the solar cell is defined as the ratio of energy output from the solar cell to input energy from the sun. Generally, the efficiency of a solar cell could be $20\%$, and a recent report shows that the $44.7\%$ efficiency has been achieved with a special technology \cite{Solar}. With regards to the wind energy,  the work in \cite{Krieger} shows that
a small wind turbine located on the roof can generate a $150$ W output.

\begin{figure*} \centering
\subfigure[Comparison of the system cost.]{ \label{fig:subfig:a}
\includegraphics[width=2.2in]{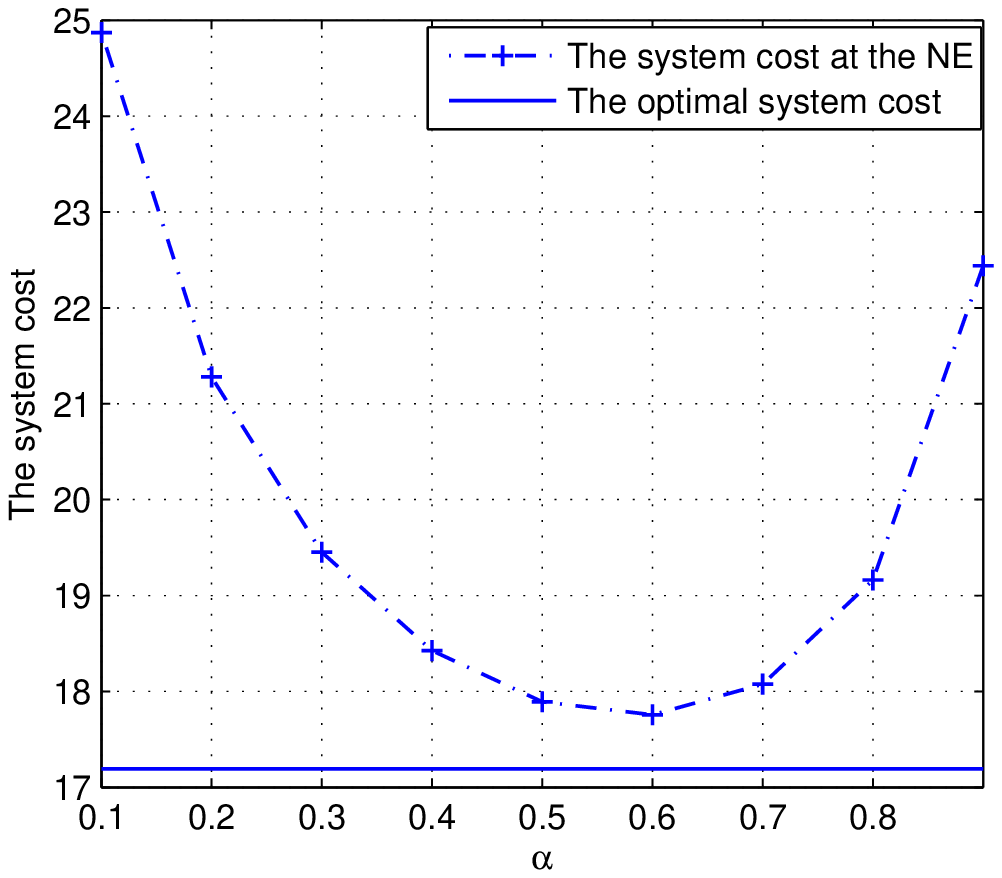}}\hspace{-0.1in}
\subfigure[The BS costs.]{ \label{fig:subfig:b}
\includegraphics[width=2.2in]{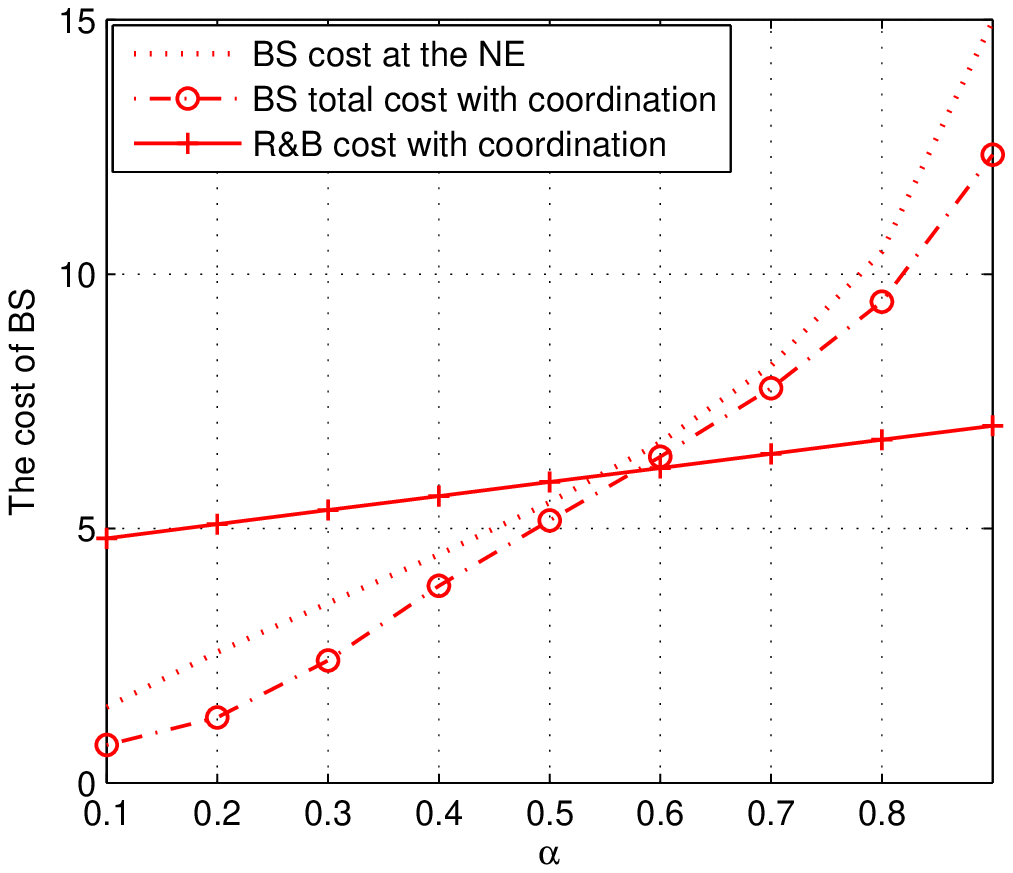}}\hspace{-0.1in}
\subfigure[The RPS costs.]{ \label{fig:subfig:b}
\includegraphics[width=2.2in]{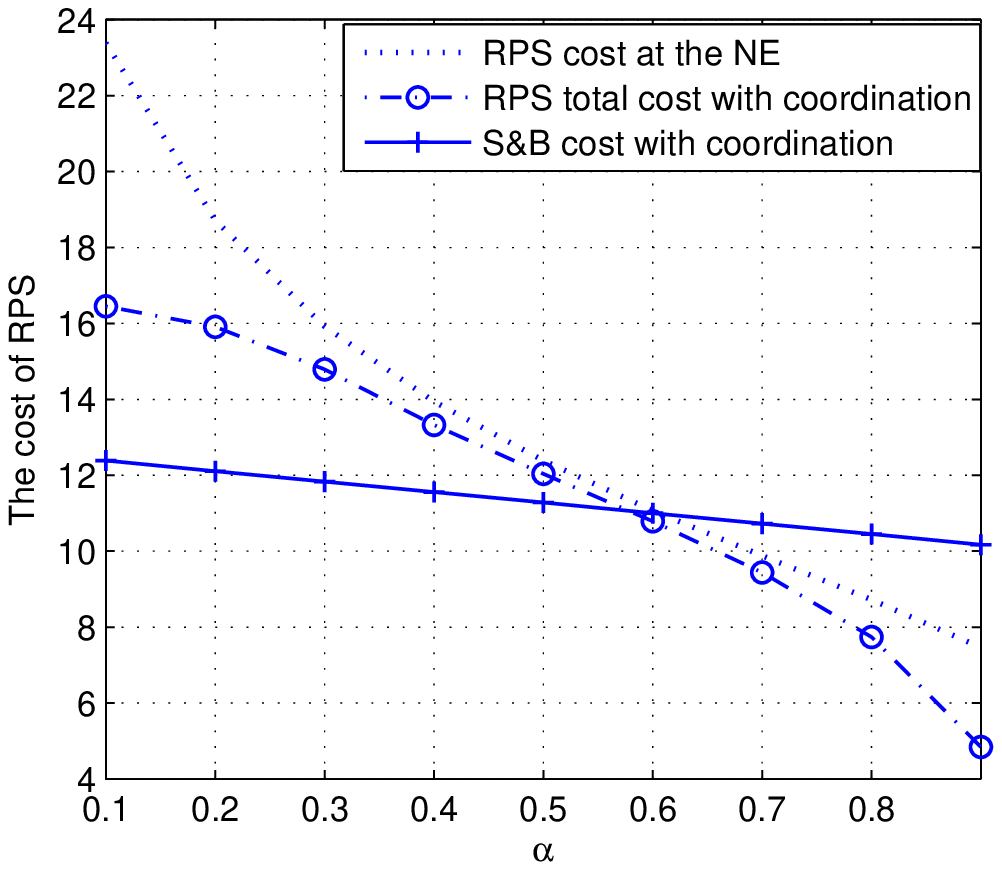}}\vspace{-0.3cm}
\caption{Equilibrium costs analysis in which R\&B represents reservation-and-backlogged, and S\&B represents supply-and-backlogged.} \label{figNEopt}
\end{figure*}

In Fig.~\ref{central2}, $\mu_0$ varies from $1.2\lambda$ to $2.2\lambda$.
Recall that the energy coefficient of each connection is $24$ W(Joule/S).
Consider $\lambda = 1.5$ connections/second as an example. Each connection accounts for one second, and consumes $240$ Joule. Then, $\mu_0 = 1.2\lambda$ means that the RPS has an average $\mu_0 = 4.32k$ Joule output in 10-second, which is $432$ W. As shown in Fig.~9, by increasing $\frac{\mu_0}{\lambda}$, the optimal system cost reduces and the optimal energy reservation level $s$ also can be reduced. In other words, if the RPS has an efficient production capacity, i.e., the solar panel experiences a clear sky, the BS can make a small reservation, and the supply cost of the RPS can be also reduced.

\begin{figure}[t]
\centering
\includegraphics[width=65mm]{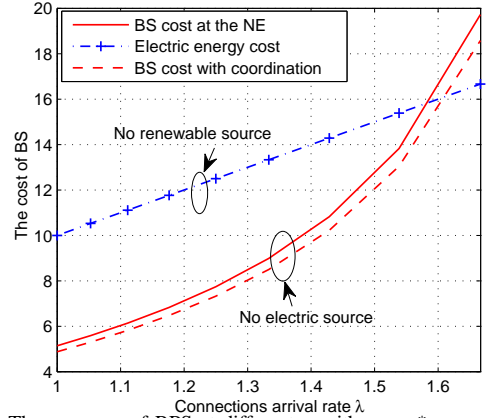}\vspace{-0.5cm}
\caption{The response of RPS at different $\rho$ with $s=s^*$.}\label{com-electric}
\end{figure}

\begin{figure}[t]\centering
\subfigure[The cost of the BS with adaptive power management.]{ \label{fig:subfig:a}
\includegraphics[width=65mm]{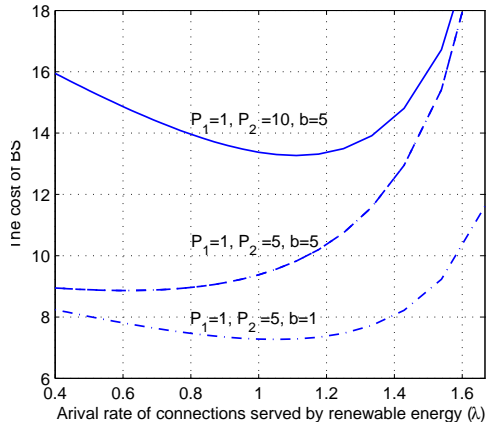}}\hspace{-0.1in}
\subfigure[The energy cost gain with the combination energy.]{ \label{fig:subfig:b}
\includegraphics[width=65mm]{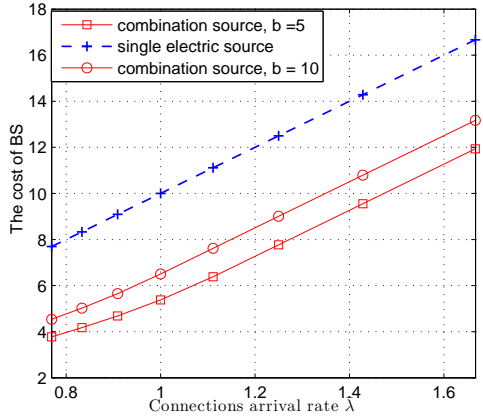}}\hspace{-0.1in}\vspace{-0.3cm}
\caption{System performance at the NE, and the energy management based on the NE.} \label{combiantion}
\vspace{-0.35em}
\end{figure}

Next, we study the decentralized situations. Set $b=10,~c_s=5,~\varphi=1$. In Fig. \ref{figNE}, we show the best responses of the BS and the RPS at the NE. Neither the BS nor the RPS can reduce their individual cost via a unilateral deviation from these strategies at the NE $(\nu^*,s^*)$. The value of the splitting factor $\alpha$ is exogenously determined. As $\alpha$ increases, the BS incurs more backlogged costs. We can observe from Fig. \ref{figNE}(a) that,  a larger $s^*$ corresponds to a larger $\alpha$. This means that the BS should reserve more energy to make up
for the backlogged cost with a larger $\alpha$. For the RPS, a larger $\alpha$ incurs less splitting backlogged cost. Thus, as shown in Fig. \ref{figNE}(b), the RPS will set a smaller energy supply rate with a larger $\alpha$. Moreover, the costs of RPS in terms of load factor $\rho$ are illustrated in Fig. \ref{figNE}(c). At $s=s^*$, $\nu=0.05\rightarrow0.8$ corresponds to
$\rho=0.56\rightarrow0.95$. The figure shows the best response of the RPS occurs at
heavy traffic situations, e.g., $\rho>0.7$. These results corroborate the intuition that a larger supply rate will result in a larger supply cost for the RPS.

In Fig.\ref{figNEopt}(a), we compare the system cost at the NE and the optimal system cost ($b=10,~c_s=5,~\varphi=1$). In this figure, we can see that the gap between such two values is quite small with the backlogged cost factor $\alpha = 0.5$. In order to coordinate the system to achieve a minimum system cost, based on the cost sharing, a kind of transfer payment contract is proposed. In Fig.\ref{figNEopt}(b) and Fig.\ref{figNEopt}(c), we show the BS cost and the RPS cost at
the NE. For illustration purpose, we set the cost sharing factor, $\varepsilon$, to be a medium value of the range. The corresponding costs with such a contract coordination are illustrated. Also, we show the
reservation and backlogged costs of the BS, and the supply and backlogged costs of the RPS in the centralized optimal solution. Moreover, the gap between the RB cost in the centralized optimal solution and the total cost of the BS with coordination is the transfer payment in the contract.

Fig. \ref{com-electric} shows the energy cost saving by using the renewable energy.
We set $\mu_0 = 2,~\alpha=0.5$. Clearly, using the renewable energy not only achieves lower cost in most cases, but also allows to avoid electric energy consumption.
In particular, if the connections arrival rate
of the BS is $1.3$ connections per second, then the saving of the total electric energy which is mostly
generated from the fossil fuel will be  $1.3\cdot24\cdot60 \cdot 1000= 0.52$ kWh in 1000 minutes.
However, if the connections' arrival rate is larger than $1.6$, i.e., the
load factor is $0.9$, then, using the renewable energy is not efficient, since
the QoS reducing cost overcomes the cost benefit of the green energy.

\begin{figure*}[t] \centering
\subfigure[The energy supply rate demand and allocation.]{\label{fig:subfig:a}
\includegraphics[width=2.2in]{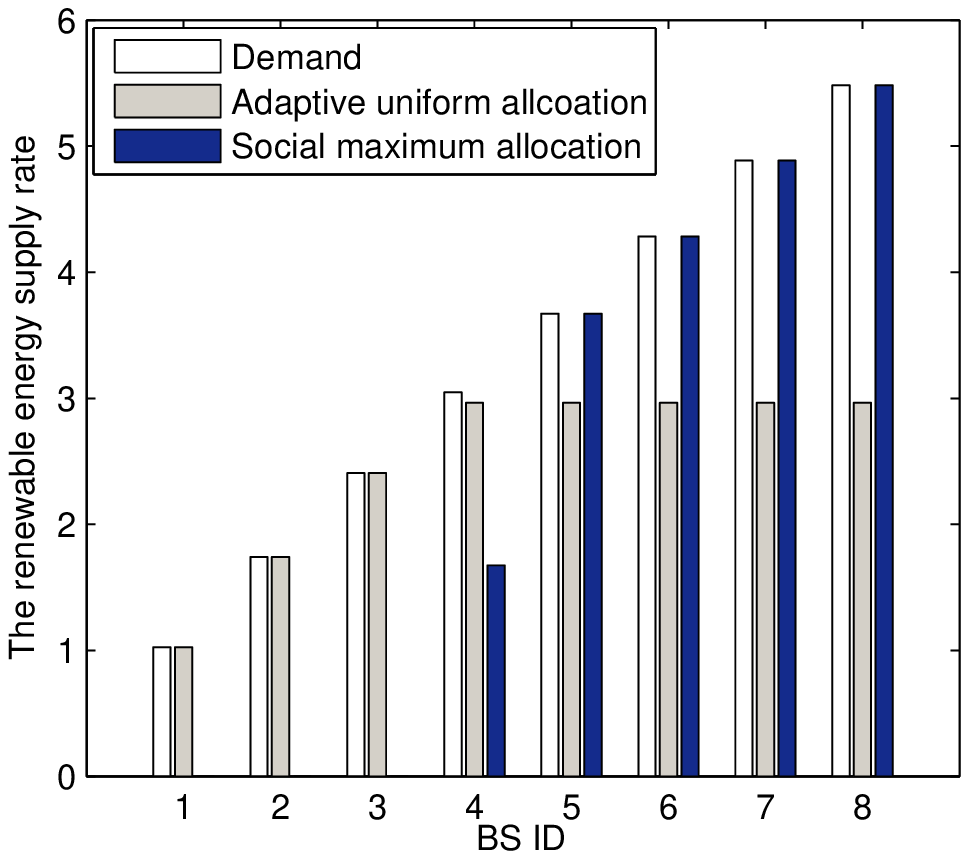}}\hspace{-0.1in}
\subfigure[Illustration of the dominant equilibrium.]{ \label{fig:subfig:b}
\includegraphics[width=2.2in]{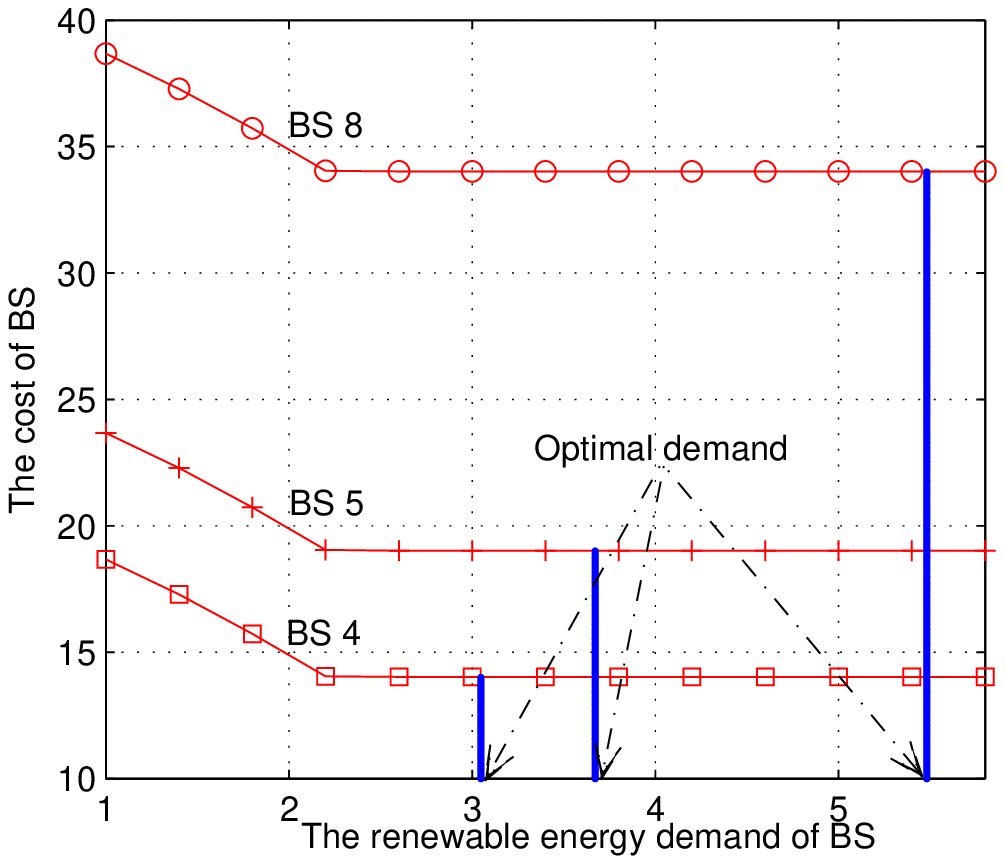}}\hspace{-0.1in}
\label{fig:subfig:b}
\subfigure[Comparison of the Pareto mechanism and the IC mechanism.]{
\includegraphics[width=2.2in]{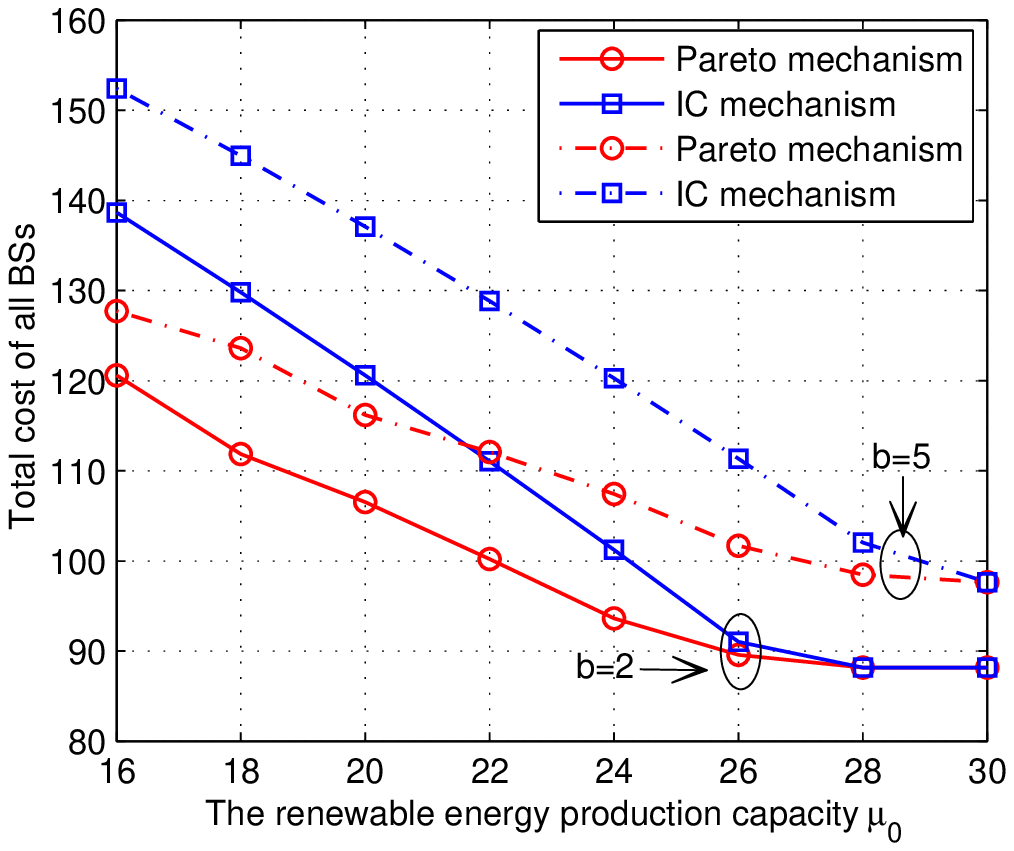}}\hspace{-0.1in}\vspace{-0.3cm}
\caption{Analysis for the energy supply allocation.} \label{allocation}
\end{figure*}
Next, the situations in which the BS powered by the combination electric and renewable sources are investigated. Different countries or provinces may have different energy prices.
In Fig. \ref{combiantion}(a),  we set $\bar{\lambda} = 1.8$, $P_1=1$, $\mu_0=2$ and $\alpha=0.5$.
As the electric grid energy price $P_2$ increases, the BS will allocate a bigger fraction of its connections to be served by the renewable energy source.
For instance, with the backlogged cost $b=5$, if $P_2=5\rightarrow P_2=10$,
to achieve the minimum cost, the BS will set $\lambda= 0.67\rightarrow\lambda=1.11$.
However, if the BS has a low QoS cost, i.e, $b= 1, P_2=5$, then $\lambda=1.05$.
This means that the BS can obtain more green energy gain with low backlogged costs.
In Fig. \ref{combiantion}(b), we compare the cost of the BS powered by the electric grid, and
the cost of the BS served by a combination energy.
We observe that,  even as connections arrival rate $\bar{\lambda}$ increases, using the green energy still benefits the BS. This is because the adaptive power management allocates a fraction of connections to be served by the electric grid, not yielding a high QoS cost.

Consider the case in which $\bar{\lambda} = 1.25$ as a example. The cost gain is about $4.5$ with $b=5$ in this case. As we stated above, this cost accounts for 1000 minutes in term of cents.
Consequently,  the cost saving for one BS is about 1.94 dollars per month.
Even though such a value is seemingly small, this cost saving can be significant when a large number of
BSs are deployed. Also, the lower use of power from the electric grid will reduce the overall $\text{CO}_2$ footprint of wireless networks.

\subsection{Analysis for the multiple BSs scenario}
In Fig. \ref{allocation}(a), we examine the renewable energy supply rate allocation with $N=8$ BSs. We set $\mu_0=20$ which represents that the maximum average energy output rate of the RPS is $20$. Let $p=2$, $P_1 = 1$, $b_i=2,i\in\mathcal{N}$ and $P_2 = 10$.
$\bar{\lambda_i}=0.5i$ where $i\in[1,8]$ is the index of BS $i$.
Fig.~\ref{allocation}(a) shows the BSs' ordered energy supply rate and the allocated rate by using the Pareto mechanism and the Proposed adaptive uniform allocation.
The figure illustrates that the Pareto mechanism favors the BS with a large amount
demand. However, the uniform allocation favors the BS with a small order, and the
BSs 5, 6, 7 and 8 get a uniform allocation.

Corresponding to the scenario in Fig. \ref{allocation}(a), Fig. \ref{allocation}(b) shows the cost of the BS when the reported demand varies. Consider BSs $4$, $5$ and $8$ as examples.
We observe that as the reported demand reduces, the cost of the BS
cannot reduces and may increases. And, when the reported demand is larger than the optimal demand, the cost of a BS will maintain a certain value. Thus, it is a dominant equilibrium for all BSs to truthfully report their optimal energy demands under the proposed adaptive uniform allocation mechanism. Moreover, note that the mechanism requires at most 1 adjustment.

In Fig. \ref{allocation}(c), we show the total BSs'cost under the Pareto mechanism and the IC mechanism.
It is clear to see that the adaptive uniformly allocation is not efficient.
However, the gap between the two mechanisms in terms of total BSs'cost becomes smaller as the energy production capacity of the RPS, $\mu_0$, increases. For instance, when $b=5$, $\mu_0=28$, the IC mechanism achieves almost the same system cost as the Pareto mechanism, since the total demand of BSs is about $28$ as well, and both the two mechanisms will allocate all of the energy capacity to BSs.
Moreover, by choosing the IC mechanism, the RPS can acquire truthful energy demand information from the BSs. In practice, the RPS can use the informative IC mechanism in some stages, yielding
the RPS's secure decision-makings on capacity planning and sales planning. Also, it can switch to
the Pareto mechanism to make a larger social welfare, and higher revenue.

\section{Conclusions}\vspace{-0.35em}
In this paper, we have studied a wireless network in which the BS  is able to acquire power from renewable energy sources. In the studied model, the BS can make
reservations for the renewable energy in order to support continuous wireless
connections.  We have formulated the problem as a noncooperative game between the BS and the RPS.
In this game, an M/M/1 make-to-stock queuing model has been used to analyze the
the competition between energy reservation strategy and the energy supply rate setting.
Several approaches for improving the efficiency of the Nash equilibrium as well as for better controlling the purchase of renewable energy have been proposed.
Then, we have extended the model to the case in which
multiple BSs operate with a monopolistic RPS. For this case,
we have proposed an allocation mechanism in which
BSs have an incentive to truthfully report their optimal demands.

Simulation results have shown that
using green energy powered BS yields a significant electric energy saving.
Furthermore, our results also have revealed that a monopolistic RPS can use the proposed informative IC mechanism to acquire the market information in the multi-BS market, and can also switch to the Pareto mechanism to achieve the optimal social welfare.
For the future work, how to
develop mathematical techniques to examine the impacts of dynamic prices on both the green energy allocation and the Qos cost of wireless networks will be very interesting and challenging.

\section*{ACKNOWLEDGEMENTS}
This work was supported by  the China National Science Foundation under Grants (61201162, 61271335, 61302100, 61471203); the U.S. National Science Foundation under Grants CNS-1460333, CNS-1406968, and AST-1506297; The National Fundamental Research Grant of China (2011CB302903); National Key Project of China (2010ZX03003-003-02); Basic Research Program of Jiangsu Province (NSF), China (BK2012434); Postdoctoral Research Fund of China (2012M511790); New Teacher Fund for Doctor Station of the Ministry of Education, China (20123223120001,20133223120002); and ``Towards Energy-Efficient Hyper-Dense Wireless Networks with Trillions of Devices'', a Commissioned Research of National Institute of Information and Communications Technology (NICT), Japan.

\bibliographystyle{IEEEtran}

\begin{IEEEbiography}[{\includegraphics[width=1in,height=1.25in,clip,keepaspectratio]{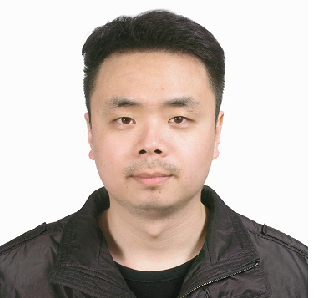}}]{Dapeng Li}(M'11) received the B.S. degree in Electronics from Harbin Engineering
University, China, in 2003, and  the M.S. degree in Communication
Systems from Harbin Engineering University, China, in 2006. He
received his Ph.D. degree from the Department of Electronic
Engineering, Shanghai Jiao Tong University, China, in 2010.
From Qct. 2007 to Aug. 2008, he have joined in the 4G/LTE development
at the R\&D Center, Huawei Company, Shanghai, China.
Starting Jan. 2011, he is a faculty member in the College of
Telecommunications and Information Engineering, Nanjing University
of Posts and Telecommunications, China. His research
interests include future cellular networks, machine to machine communications,   mobile ad hoc networks,
cognitive radio networks, radio resource management and cooperative
communications.
\end{IEEEbiography}
\begin{IEEEbiography}[{\includegraphics[width=1.2in,height=1.25in,clip,keepaspectratio]{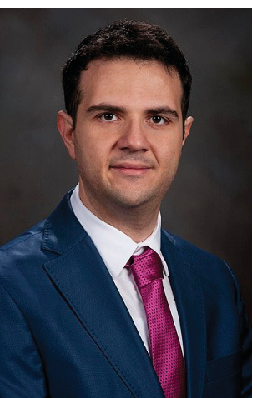}}]{Walid Saad}(S'07, M'10) received his B.E. degree in Computer and Communications Engineering from the Lebanese  University, in 2004, his M.E. in Computer and Communications Engineering from the American  University of Beirut (AUB) in 2007, and his Ph.D degree from the University of Oslo in 2010.

Currently,  he is an Assistant Professor at the Bradley Department of Electrical and Computer Engineering at Virginia Tech, where he leads the Network Science, Wireless, and Security (NetSciWiS) laboratory, within the Wireless@VT research group. Prior to joining VT, he was a faculty at the Electrical and Computer Engineering Department at the University of Miami and he has held several research positions at  institutions such as Princeton University and the University of Illinois at Urbana-Champaign.

His  research interests include wireless and social networks, game theory, cybersecurity, smart grid, network science,  cognitive radio, and self-organizing networks.  He has published one textbook and over 115 papers in these areas. Dr. Saad is the recipient of the NSF CAREER award in 2013 and of the AFOSR summer faculty fellowship in 2014. He was the author/co-author of three conference best paper awards at WiOpt in 2009, ICIMP in 2010, and IEEE WCNC in 2012. He currently serves as an editor for the IEEE Transactions on Communications and the IEEE Communication Tutorials \& Surveys.
\end{IEEEbiography}

\begin{IEEEbiography}[{\includegraphics[width=1in,height=1.25in,clip,keepaspectratio]{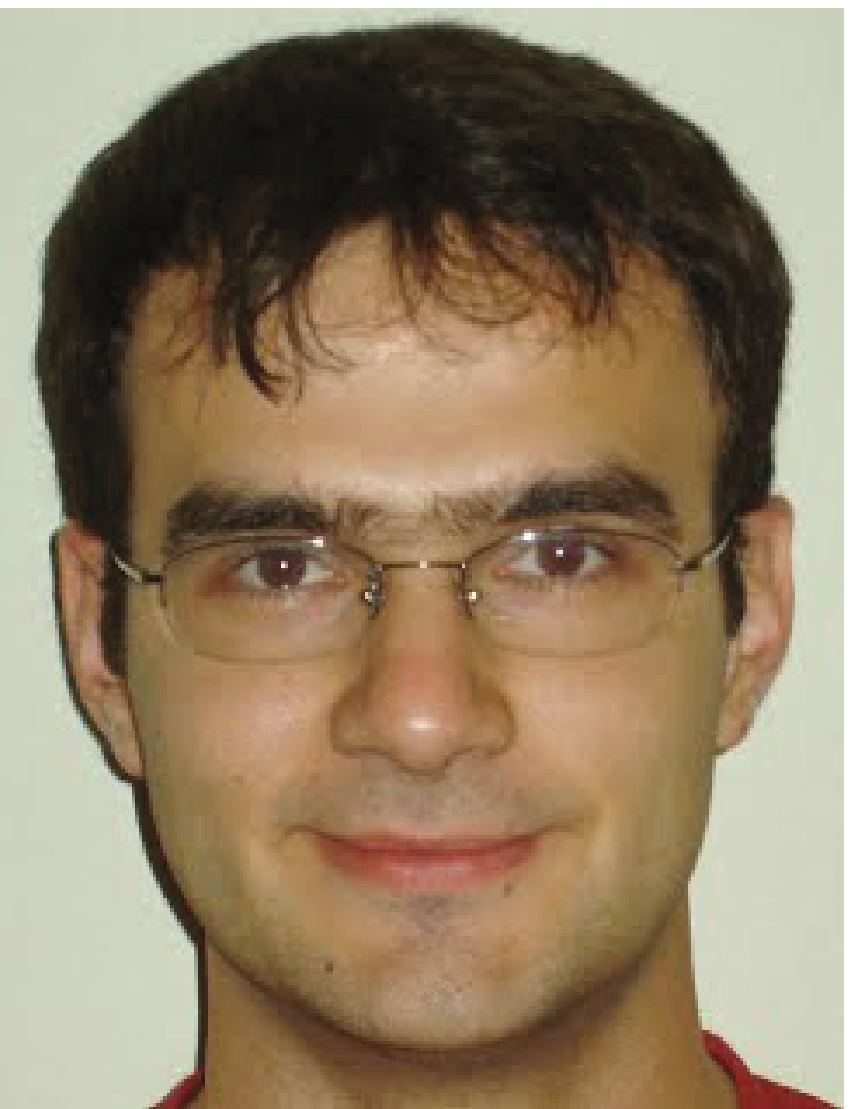}}]{Ismail Guvenc}
(senior member, IEEE) received his Ph.D. degree in electrical engineering from University of South Florida in 2006, with an outstanding dissertation award. He was with Mitsubishi Electric Research Labs during 2005, and with DOCOMO Innovations Inc. between 2006-2012, working as a research engineer. Since August 2012, he has been an assistant professor with Florida International University.
His recent research interests include heterogeneous wireless networks and future radio access beyond 4G wireless systems. He has published more than 100 conference/journal papers and book chapters, and several standardization contributions. He co-authored/co-edited three books for Cambridge University Press, is an editor for IEEE Communications Letters and IEEE Wireless Communications Letters, and was a guest editor for four special issue journals/magazines on heterogeneous networks. Dr. Guvenc is an inventor/coinventor in 23 U.S. patents, and has another 4 pending U.S. patent applications. He is a recipient of the 2014 Ralph E. Powe Junior Faculty Enhancement Award and 2015 NSF CAREER Award.
\end{IEEEbiography}

\begin{IEEEbiography}[{\includegraphics[width=1in,height=1.25in,clip,keepaspectratio]{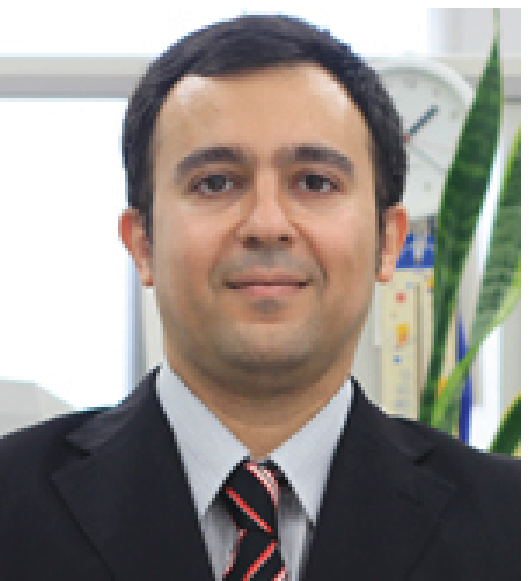}}]{Abolfazl Mehbodniya} received his Ph.D degree from the National Institute of Scientific Research-Energy, Materials, and Telecommunications (INRS-EMT), University of Quebec, Montreal, QC, Canada in 2010. From 2010 to 2012 he was a JSPS postdoctoral fellow at Tohoku University. Since Jan 2013, he has been an assistant professor at department of communications engineering, Tohoku University. Dr. Mehbodniya has 10+ years of experience in electrical engineering, wireless communications, and project management. He has over 40 published conference and journal papers in the areas of radio resource management, sparse channel estimation, interference mitigation, short-range communications, 4G/5G design, OFDM, heterogeneous networks, artificial neural networks (ANNs) and fuzzy logic techniques with applications to algorithm and protocol design in wireless communications. He is the recipient of JSPS fellowship for foreign researchers, JSPS young faculty startup grant and KDDI foundation grant.
\end{IEEEbiography}

\begin{IEEEbiography}[{\includegraphics[width=1in,height=1.25in,clip,keepaspectratio]{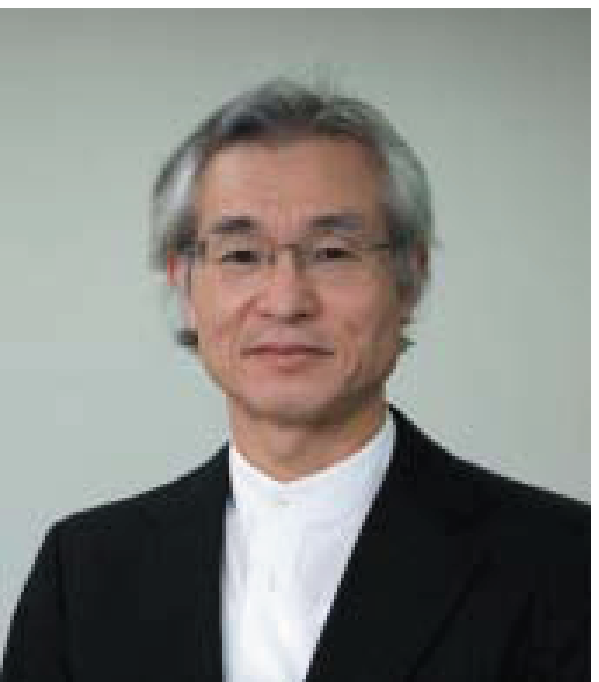}}]{Fumiyuki Adachi}
received the B.S. and Dr. Eng. degrees in electrical engineering from Tohoku University, Sendai, Japan, in 1973 and 1984, respectively. In April 1973, he joined the Electrical Communications Laboratories of NTT and conducted various types of research related to digital cellular mobile communications. From July 1992 to December 1999, he was with NTT DoCoMo, where he led a research group on Wideband CDMA for 3G systems. Since January 2000, he has been with Tohoku University, Sendai, Japan, where he is a Distinguished Professor of Communications Engineering at the Graduate School of Engineering. Professor Adachi is an IEEE Fellow and an IEICE Fellow. He is a pioneer in wireless communications since 1973 and has largely contributed to the design of wireless networks from 1 generation (1G) to 4G. He is listed on ISIHighlyCited.com and is an IEEE Vehicular Technology Society Distinguished Lecturer since 2012. He is a vice president of IEICE Japan since 2013. He was a recipient of the IEEE Vehicular Technology Society Avant Garde Award 2000, IEICE Achievement Award 2002, Thomson Scientific Research Front Award 2004, Ericsson Telecommunications Award 2008, Telecom System Technology Award 2010, Prime Minister Invention Award 2010, and KDDI Foundation Excellent Research Award 2012. His research interests include wireless signal processing for wireless access, equalization, transmit/receive antenna diversity, MIMO, adaptive transmission, channel coding, and wireless systems.
\end{IEEEbiography}
\end{document}